\documentclass[draftclsnofoot,onecolumn,12pt,romanappendics]{IEEEtran}
\usepackage{amsmath}
\usepackage{mathrsfs}
\usepackage{pifont}
\usepackage{bbding}
\usepackage{amsmath,epsfig}
\usepackage{stmaryrd}
\usepackage{amssymb}
\usepackage{amsfonts}
\usepackage{epic}
\usepackage{graphicx}
\usepackage{curves}
\usepackage{cite}

\usepackage{algorithm}

\usepackage{algpseudocode}

\usepackage{stfloats}
\usepackage{latexsym}
\usepackage{epstopdf}
\usepackage{epic}
\usepackage{bm}
\usepackage{multirow}
\usepackage{xcolor}
\begin{document}
\newcommand{\reffig}[1]{Fig. \ref{#1}}
\newcommand{\figref}[1]{\figurename~\ref{#1}}

\newtheorem{definition}{Definition}
\newtheorem{theorem}{Theorem}
\newtheorem{corollary}{Corollary}
\newtheorem{proposition}{Proposition}
\newtheorem{lemma}{Lemma}
\newtheorem{property}{Property}
\newtheorem{remark}{Remark}
\renewcommand{\algorithmicrequire}{\textbf{Input:}}  
\renewcommand{\algorithmicensure}{\textbf{Output:}}  

\title{Energy-Efficient Resource Allocation for  Wireless Powered Communication Networks   }

\author{\IEEEauthorblockN{Qingqing Wu, \emph{Student Member, IEEE},  Meixia Tao, \emph{Senior Member, IEEE}, \IEEEauthorblockN{Derrick Wing Kwan Ng},  \emph{Member, IEEE}, Wen Chen, \emph{Senior Member, IEEE}, and Robert Schober, \emph{Fellow, IEEE}
\thanks{Qingqing Wu, Meixia Tao, and Wen Chen  are with Department of Electronic Engineering, Shanghai Jiao Tong University, email: \{wu.qq, mxtao, wenchen\}@sjtu.edu.cn. Derrick Wing Kwan Ng and Robert Schober are with the Institute for Digital Communications,
Friedrich-Alexander-University Erlangen-N¨urnberg, email: \{wingn, rschober\}@ece.ubc.ca. This paper has been presented in part at IEEE ICC 2015 \cite{qing15_wpcn}. }}  }

\maketitle
\begin{abstract}
This paper considers a wireless powered communication network (WPCN), where multiple users harvest energy from a dedicated power station and then communicate with an information receiving station. Our goal is to investigate the maximum achievable energy efficiency (EE) of the network via joint time allocation and power control while taking into account the initial battery energy of each user.
 We first study the EE maximization problem in the WPCN without any system throughput requirement. We show that the EE maximization problem for the WPCN can be cast into EE maximization problems for two simplified networks via exploiting its special structure. For each problem,  we derive the optimal solution and provide the corresponding physical interpretation,  despite the non-convexity of the problems.  Subsequently, we study the EE maximization problem under a minimum system throughput constraint. Exploiting fractional programming theory, we  transform the resulting non-convex problem into a standard convex optimization problem.  This allows us to characterize the optimal solution structure of joint time allocation and power control and to derive an efficient iterative algorithm for obtaining the optimal solution.
 Simulation results verify our theoretical findings and demonstrate  the effectiveness of the proposed joint time and power optimization.

\end{abstract}

\section{Introduction}
%

 Energy harvesting allows devices to harvest energy  from ambient sources, and has attracted considerable attention in both academia and industry \cite{qing15_wpcn,krikidis2014simultaneous}.  {Energy harvesting from natural renewable sources, such as solar and wind, can provide a green and renewable energy supply for wireless communication systems.
 However, due to the intermittent nature of renewable energy sources, the energy collected at the receiver is not controllable, and the communication devices may not always be able to harvest sufficient energy.
 On the other hand, it has been shown that wireless receivers can also harvest  energy from radio frequency (RF) signals, which is known as wireless energy transfer (WET) \cite{qing15_wpcn,krikidis2014simultaneous}. Since the RF signals are generated by dedicated devices, this type of energy source is more stable than natural renewable sources.}

Two different lines of research can be identified in WET. 
{The first line focuses on simultaneous  wireless information and power transfer (SWIPT), where the wireless devices are able to split the received signal into two parts, one for information decoding and the other one for energy harvesting \cite{zhangrui13_mimo,derrick_wipt,xunzhou14_ofdm,kwan2013robust,NgLS14}.  SWIPT has been studied for example for multiple-input multiple-output (MIMO) \cite{zhangrui13_mimo}, multiuser orthogonal frequency division multiplexing access (OFDMA) \cite{derrick_wipt,xunzhou14_ofdm}, multiuser multiple-input single-output  (MISO) \cite{kwan2013robust},  and cognitive radio \cite{NgLS14}. These works generally consider the power splitting ratio at the receiver side to study the fundamental tradeoff between the achievable throughput and the harvested energy.
The second line of research in WET pursues wireless powered communication networks (WPCNs), where the wireless devices are first powered by WET and then use the harvested energy to transmit data signals
\cite{huang14enabling,ju14_throughput,ju2014user,ju2014optimal}.}
 In \cite{ju14_throughput},  the downlink (DL) WET time and the uplink (UL) wireless information transmission (WIT) time are jointly optimized to maximize the system throughput. \textcolor{blue}{Then, WPCNs with user cooperation and full-duplex, relay,  multi-antenna, massive MIMO, and cognitive techniques are further studied in \cite{ju2014user,ju2014optimal,chen2015harvest,liuliang14_wpcn,yang2014throughput,lee2015cognitive,che2014spatial}, respectively.}
Moreover, the authors in \cite{ding2014power} investigate how an energy harvesting relay can distribute its harvested energy to support the communication of  multiple source-destination pairs.
However, most existing works on WET aim to  improve the system throughput while neglecting the energy utilization efficiency which is also a critical issue for next generation communication systems, especially for energy harvesting based systems \cite{wu2012green,li2014energy,xiaomingchen_13,miao2013energy1,qing1,huang2013simultaneous}.


Because of the rapidly rising energy costs and the tremendous carbon footprints of existing systems \cite{wu2012green}, energy efficiency (EE), measured in bits per joule, is gradually accepted as an important design criterion for future communication systems \cite{qing15_scheduling,xiong2012,li2014energy}.
   The authors in \cite{derrick_wipt} study the resource allocation for EE maximization in SWIPT for OFDMA systems requiring minimum harvested energy guarantees for multiple receivers. However, the conclusions and proposed methods in \cite{derrick_wipt}  are not applicable to the WPCN scenario due to the fundamentally different system architecture.  Energy-efficient power allocation for large-scale MIMO systems is  investigated in \cite{xiaomingchen_13}. Yet, the resource allocation is optimized only for the single-receiver scenario and  cannot be directly extended to the multiuser case due to the coupling between  time allocation and power control.  Moreover,  the circuit power consumption of the user terminals is ignored in \cite{zhangrui13_mimo,derrick_wipt,xunzhou14_ofdm,kwan2013robust,NgLS14,ju14_throughput,ju2014user,ju2014optimal}.  However, as pointed out in \cite{kim2010leveraging}, the circuit power consumption is non-negligible compared to the power consumed for data transmission,  especially for  small scale and short range applications. { Furthermore, in the WPCN,  energy is not only consumed in the UL WIT stage but also in the DL WET stage during which no data is transmitted. In fact, a significant amount of energy may be consumed during DL WET in order to combat the wireless channel attenuation. Therefore, EE optimization is even more important in WPCN than in traditional wireless communication networks.}


 In this paper,  we consider the WPCN where multiple users first harvest energy from a power station and then use the harvested energy to transmit signals to an information receiving station. 
 The considered system model is most closely related to that in \cite{ju14_throughput}. However, there are three important differences.
 {First, a hybrid station is employed in \cite{ju14_throughput}, i.e., the power station for WET and the information receiving station for WIT are co-located. Hence, a user near the hybrid station enjoys not only higher WET channel gain in the DL but also higher WIT channel gain in the UL compared with users that are far from the hybrid station. This phenomenon is referred to as ``doubly near-far'' problem in \cite{ju14_throughput}. To avoid this problem, in this paper, the information receiving station is not restricted to be co-located with the power station. Hence, a user far from the power station can be near the information receiving station and visa versa.}
  Second, in contrast to \cite{ju14_throughput}, each user is equipped with a certain amount of initial energy and can store the harvested energy from the current transmission block for future use. This generalization provides users a higher degree of flexibility in utilizing the harvested energy and improves thereby the EE of practical communication systems.
 Third, unlike \cite{ju14_throughput},  we focus on maximizing the system EE while guaranteeing a minimum required system throughput instead of maximizing the system throughput.
The main contributions  and results of this paper can be summarized as follows:
\begin{itemize}
\item {We formulate the EE maximization problem for multiuser WPCN with joint time allocation and power control. Thereby, we explicitly take into account the circuit energy consumption of the power station and the user terminals. In the first step, we investigate the system EE of WPCN providing best-effort communication, i.e., WPCN that do not provide any system throughput guarantee.
  Subsequently, to meet the QoS requirements of practical systems, the EE maximization problem is studied for the case with a minimum required system throughput.}

\item { For the case of best-effort communication, we reveal that the energy-efficient WPCN are equivalent to either the network in which the users are only powered by the initial energy, i.e., no WET is exploited, or the network in which the users are only powered by WET, i.e., no initial energy is used. We refer to the former type of network as ``initial energy limited communication network'' (IELCN) and to the latter type of networks as ``purely wireless powered communication network'' (PWPCN). }
For the IELCN, we show that the most energy-efficient transmission strategy is to schedule only the user who has the highest user EE. In contrast, for the PWPCN, we find that:  1) in the WET stage, the power station always transmits with its  maximum power;  2) it is not necessary for  all users  to transmit signals in the WIT stage, but all scheduled users will deplete all of their energy; 3) the maximum system EE  can always be achieved by occupying all available time. Based on these observations, we derive a closed-form expression for the system EE based on the user EEs,  which transforms the original problem into a user scheduling problem that can be solved efficiently.
  \item For the case of throughput-constrained WPCN, exploiting fractional programming theory, we  transform the original problem into  a standard convex optimization problem. Through the analysis of the Karush-Kuhn-Tucker (KKT) conditions, we characterize the optimal structure of time allocation and power control, and propose an efficient iterative algorithm to obtain the optimal solution. We show that for a sufficiently long transmission time, the system EE is maximized by letting each user achieve its own maximum user EE. For a short transmission time, users can only meet the minimum system throughput requirement at the cost of sacrificing system EE.
 \end{itemize}

The  remainder of this paper is organized as follows. Section II  introduces some preliminaries regarding WPCN. In Section III, we study best-effort communication in energy-efficient WPCN. In Section IV, we investigate the EE maximization problem in the presence of a minimum system throughput requirement. Section V provides extensive simulation results to verify our analytical findings and the paper is concluded in Section VI.
\section{System Model and Preliminaries}
\subsection{System Model}
\begin{figure}[!t]
\centering
\includegraphics[width=4in]{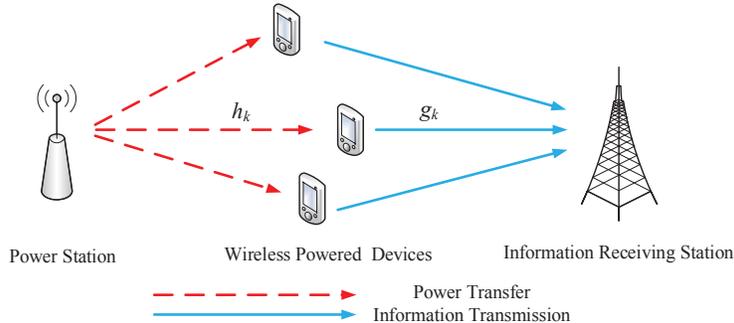}
\caption{The system model  of a multiuser wireless powered communication network. }\label{convergence}
\end{figure}
 We consider a  WPCN, which consists of one power station,  $K$ wireless-powered users, denoted by $U_k$, for $k=1, ..., K$,  and one information receiving station,  that is not necessarily co-located with the power station, as illustrated in Fig. 1. {As a special case,  the information receiving station and the power station may be integrated into one hybrid station as suggested in \cite{ju14_throughput}, which leads to lower hardware complexity but gives rise to the ``doublely near-far'' problem.}
  The ``harvest  and then transmit" protocol is employed for the WPCN. Namely, all users first harvest energy from the RF signal broadcasted  by the power station in the DL, and then transmit the information signal  to the information receiving
station in the UL \cite{ju14_throughput}.   For simplicity of implementation,  the power station, the information receiving station, and all users are equipped with a single antenna and operate in the time division mode over the same frequency band \cite{derrick_wipt,ju14_throughput}.  To be more general,  we assume that user $k$, for $k=1, ..., K$, is equipped with a rechargeable built-in battery  with an initial energy of $Q_k$ (Joule). The initial energy may be the energy harvested and stored in previous transmission blocks. This energy can be used  for WIT in the current block.

Assume that both the DL and the UL channels are quasi-static block fading channels.
 The DL channel power gain between the power station  and user terminal $k$ and the UL channel power gain between  user terminal $k$ and the information receiving  station are denoted as $h_k$ and $g_k$, respectively.  Note that both $h_k$ and $g_k$ capture the joint effect of  path loss,  shadowing, and multipath fading.  We also assume that the channel state information (CSI) is perfectly known at the power station as we are interested in obtaining an EE upper bound for practical WPCN \cite{ju14_throughput}.   Once calculated, the resource allocation policy is  then sent to the users to perform energy-efficient transmission. We assume that the energy consumed for estimating and exchanging CSI can be drawn from a dedicated battery which does not rely on the harvested energy \cite{huang2013simultaneous}.
 { We note that signaling overhead and imperfect CSI will result in a performance degradation but the study of  their impact on the system EE  is beyond the scope of this paper. For a detailed treatment of CSI acquisition in WPCN, we refer to \cite{xu14,zeng15}.}


During the WET stage, the power station broadcasts an RF signal for a time duration $\tau_0$ at a transmit power $P_0$.
{The energy harvested from channel noise and the received UL WIT signals
from other users is assumed to be negligible,  since the noise power is generally much smaller than the received signal power and the transmit powers of the users are much smaller than the transmit power of the power station in practice \cite{ju14_throughput,ju2014user,ju2014optimal,shi2013joint}.}
Thus, the amount of energy harvested at  $U_k$ can be modeled as
\begin{align}\label{eq3}
E^h_k=\eta \tau_0 P_0 h_k,  ~~~~~k=1,\cdots, K,
\end{align}
where $\eta \in (0,1]$ is the energy conversion efficiency 
which depends on the type of receivers \cite{ju14_throughput}.

During the WIT stage, each user $k$ transmits an independent information signal to the receiving station in a time division manner at a transmit power $p_k$.
  Denote  the information transmission time of user $k$ as $\tau_k$. Then, the achievable throughput of $U_k$  can be expressed as
\begin{align}\label{eq6}
B_k=\tau_k W \log_2\left(1+\frac{p_kg_k}{\Gamma\sigma^2}\right),
\end{align}
where $W$ is the bandwidth of the considered system, $\sigma^2$ denotes the noise variance, and $\Gamma$ characterizes the gap between the achievable rate and the channel capacity due to the use of practical modulation and coding schemes. In the sequel, we use $\gamma_k= \frac{g_k}{\Gamma\sigma^2}$ to denote the equivalent channel to noise ratio for WIT.
Thus,  the total throughput of the WPCN, denoted as $B_{\rm{\rm{tot}}}$, is given by
\begin{align}\label{eq7}
B_{\rm{\rm{tot}}}=\sum_{k=1}^{K}B_k=\sum_{k=1}^{K} \tau_k W \log_2(1+{p_k\gamma_k}).
\end{align}

\subsection{Power Consumption Model}
The total energy consumption of the considered WPCN consists of two parts: the energy consumed during WET and WIT, respectively. For each part, we adopt the  energy consumption model in \cite{derrick_wipt,miao2013energy1,qing1,huang2013simultaneous}, namely,   the power consumption of a transmitter includes not only the over-the-air transmit power but also the circuit power consumed for hardware processing. On the other hand, according to \cite{xu2014throughput,kim2010leveraging}, the energy consumption when users do not transmit, i.e., when they are in the idle mode as opposed to the active mode, is negligible.

During the WET stage, the system energy  consumption, denoted as $E_{\rm{WET}}$,  is modeled as
\begin{align} \label{eq8}
E_{\rm{WET}} =\frac{P_0}{\xi}\tau_{0}- \sum_{k=1}^{K}E^h_k+P_c\tau_0,
\end{align}
{where $\xi  \in (0,1]$ is the power amplifier (PA) efficiency and  $P_c$ is the constant circuit power consumption of the power station accounting for antenna circuits, transmit filter, mixer, frequency synthesizer, and digital-to-analog
converter, etc.}
In (\ref{eq8}), $P_c\tau_0$ represents the circuit energy consumed by the power station during DL WET.
Note that ${P_0}\tau_{0}- \sum_{k=1}^{K}E^h_k$ is the energy loss due to wireless channel propagation, i.e., the amount of energy that is emitted by the power station but not harvested by the users. {In practice, ${P_0}\tau_{0}- \sum_{k=1}^{K}E^h_k=  P_0\tau_{0}\left( {1}- \sum_{k=1}^{K}\eta h_k\right)$ is always positive due to the law of energy conservation and $0<\eta\leq1$ \cite{derrick_wipt,xiaomingchen_13}.}

During the WIT stage, each user independently transmits its own signal with transmit power $p_k$ during time $\tau_k$. Thus, the energy consumed by  $U_k$ can be modeled as
\begin{align} \label{eq8_1}
E_k= \frac{p_k}{\varsigma}\tau_k+p_{c}\tau_k,
\end{align}
{where $\varsigma$ and $p_{\rm{c}}$ are the PA efficiency and the circuit power consumption of the user terminals, respectively, which are assumed to be identical for all users without loss of generality.} In practice,  $E_k$ has to satisfy $E_k\leq E^h_k+Q_k$, which is known as the energy causality constraint in energy harvesting systems \cite{derrick_wipt,xiaomingchen_13}.

Therefore, the total energy consumption of the whole system, denoted as $E_{\rm{\rm{tot}}}$, is given by
\begin{align}\label{eq8_2}
E_{\rm{\rm{tot}}}=E_{\rm{WET}}+ \sum_{k=1}^{K}E_k.
\end{align}
\subsection{User Energy Efficiency}
In our previous work \cite{qing15_scheduling}, we  introduced the concept of user EE and it was shown to be directly connected to the system EE. In this subsection, we review the definition of user EE in the context of WPCN.

 \emph{Definition 1}  (User Energy Efficiency): The EE  of  user $k$,  $k=1,\cdots,K$,  is defined as the ratio of its  achievable throughput and its consumed energy in the WIT stage, i.e.,
\begin{eqnarray}\label{eq11}
ee_{k}= \frac{B_k}{E_k}=\frac{\tau_kW\log_{2}\left(1+p_k\gamma_k\right)}{\tau_k\frac{p_{k}}{\varsigma}+\tau_kp_c}=
\frac{W\log_{2}\left(1+p_k\gamma_k\right)}{\frac{p_{k}}{\varsigma}+p_c},
\end{eqnarray}
where the energy consumption includes the energy consumed in both the PA and the electronic circuits.
{Hence, $ee_k$ represents the energy utilization efficiency of user $k$ in WPCN.}

 {It can be shown that $ee_{k}$ is a strictly quasiconcave function of $p_{k}$ and has a unique stationary point which is also the maximum point \cite{Boyd}.
Therefore, by setting the derivative of $ee_{k}$ with respect to $p_{k}$ to zero, we obtain
\begin{align}
\frac{d ee_k}{d p_k}=\frac{       \frac{W\gamma_k}{(1+p_k\gamma_k)\ln2}(\frac{p_k}{\varsigma}+p_c)  - W\log_2(1+p_k\gamma_k)\frac{1}{\varsigma}       }{      \left( \frac{p_k}{\varsigma}+p_c\right)^2                }=0.
\end{align}
After some straightforward manipulations, the optimal transmit power can be expressed as
\begin{align}\label{eq12}
p^{\star}_{k}=\left[\frac{W\varsigma}{ ee^{\star}_{k}\ln2}-\frac{1}{\gamma_{k}}\right]^+, \forall\, k,
\end{align}
where  $[x]^+\triangleq \max\{x,0\}$ and $ee_k^{\star}$ is  the maximum EE of user $k$ in (\ref{eq11}).
Based on (\ref{eq11}) and (\ref{eq12}), the numerical values of  ${ee}^{\star}_{k}$ and $p_{k}^{\star}$  can be easily obtained using the bisection method \cite{Boyd}.} {As shown in the sequel, the user EE plays an important role in deriving an analytical expression for the maximum system EE as well as for interpreting the obtained expression.}

\section{Energy-Efficient Resource Allocation For Best-Effort WPCN}
In this section, we study the resource allocation in best-effort WPCN with the objective to maximize the system EE, which is defined as the ratio of the achieved system throughput to the consumed  system energy, i.e.,
$EE=\frac{B_{\rm{tot}}}{E_{\rm{tot}}}.$
 Specifically,  our goal  is to jointly optimize the \emph{time allocation} and \emph{power control} in the \emph{DL} and the \emph{UL}  for maximizing the system EE.
The system EE maximization can be formulated  as
\begin{align}\label{eq10}
EE^*= \mathop {\max }\limits_{P_{0},\tau_0,\{p_{k}\}, \{\tau_{k}\}, \forall\, k}~~ &\frac{\sum_{k=1}^{K} \tau_kW\log_2\left(1+p_k\gamma_k\right)}
{{P_{0}}\tau_0(\frac{1}{\xi}-\sum_{k=1}^{K}\eta h_k )+P_c\tau_0 + \sum_{k=1}^{K}(\frac{p_k}{\varsigma}\tau_k+p_{c}\tau_k)} \nonumber \\
\text{s.t.} ~~~~~~~&  \text{C1:}~~ P_{0}\leq P_{\mathop{\max}}, \nonumber \\
& \text{C2:}~~\frac{p_{k}}{\varsigma}\tau_{k}+p_{c}\tau_k \leq \eta P_{0}\tau_{0}h_k+Q_k, ~ \forall\, k,  \nonumber \\
&\text{C3:}~~\tau_{0}+\sum_{k=1}^{K}\tau_k\leq T_{\mathop{\max}}, \nonumber \\
& \text{C4:}~~\tau_{0}\geq0, ~  \tau_k\geq  0, ~\forall\, k,    \nonumber\\
& \text{C5:}~~P_{0}\geq0, ~ p_k\geq  0, ~\forall\, k.
\end{align}
where $EE^*$ is the maximum system EE of WPCN.
In problem (\ref{eq10}),  constraint C1 limits the DL transmit power of the power station to $P_{\mathop{\max}}$.
C2 ensures that the energy consumed for WIT in the UL does not exceed the  total available  energy which is comprised of both the harvested energy $\eta P_{0}\tau_{0}h_k$ and the initial energy $Q_k$.
In C3, $T_{\mathop{\max}}$  is the total available transmission time for the considered time block.
 C4 and C5 are non-negativity constraints on the time allocation and power control  variables, respectively.
Note that problem (\ref{eq10}) is neither convex nor  quasi-convex  due to the fractional-form objective function and the coupled optimization variables.
In general, there is no standard method for solving non-convex optimization problems efficiently. Nevertheless, in the following, we show that the considered problem can be efficiently solved by exploiting the fractional structure of the objective function in (\ref{eq10}).

\subsection{Equivalent Optimization Problems}
First, we  show that the EE maximization problem for WPCN is equivalent to two optimization problems for two simplified sub-systems. 
To facilitate the presentation, we define   $\Phi_{\mathcal{P}}$ and $\Phi_{\mathcal{I}}$ as the set of users whose initial energy levels are zero and strictly positive, respectively, i.e., $Q_k=0$  for $k\in \Phi_{\mathcal{P}}$ and $Q_k>0$ for $k\in \Phi_{\mathcal{I}}$. $\Phi_{\mathcal{P}}=\{k |Q_k=0\}$

\begin{theorem}\label{theorem0}
Problem (\ref{eq10}) is equivalent to one of the following two problems:

  {1)  The EE maximization in the pure WPCN (PWPCN) (i.e., the system where DL WET is used and only the users in $\Phi_{\mathcal{P}}$ are present for UL WIT):
   \begin{align}\label{eq132}
EE^*_{\rm{PWPCN}} \triangleq \mathop {\max }\limits_{P_{0},\tau_0,\{p_{k}\}, \{\tau_{k}\}, k\in \Phi_{\mathcal{P}} } ~~&  \frac{\sum_{k\in \Phi_{\mathcal{P}}} \tau_kW\log_2\left(1+p_k\gamma_k\right)}
{{P_{0}}\tau_0(\frac{1}{\xi}-\sum_{k=1}^{K}\eta h_k )+P_c\tau_0 + \sum_{k\in \Phi_{\mathcal{P}}}(\frac{p_k}{\varsigma}\tau_k+p_{c}\tau_k)}  \nonumber \\
\text{s.t.} ~~~~ &\text{C1:}~~ P_{0}\leq P_{\mathop{\max}},  \nonumber\\
~~& \text{C2:}~~ \frac{p_{k}}{\varsigma}\tau_{k}+p_{c}\tau_k \leq \eta P_{0}\tau_{0}h_k, ~~ k\in \Phi_{\mathcal{P}},  \nonumber\\
~~&  \text{C3:}~~\tau_{0}+\sum_{ k\in \Phi_{\mathcal{P}}}\tau_k\leq T_{\max},  \nonumber \\
~~& \text{C4:}~~\tau_{0}\geq0, ~  \tau_k\geq  0, ~k\in \Phi_{\mathcal{P}},    \nonumber\\
~~& \text{C5:}~~P_{0}\geq0, ~ p_k\geq  0, ~k\in \Phi_{\mathcal{P}},
\end{align}
where $EE^*_{\rm{PWPCN}}$ is used to denote the maximum system EE  of the PWPCN.}

{ 2) The EE maximization in the initial energy limited communication network (IELCN) (i.e., the system where DL WET is not used and only the users in $\Phi_{\mathcal{I}}$ are present for UL WIT):
  \begin{align}\label{eq133}
EE^*_{\rm{IELCN}}\triangleq\mathop {\max }\limits_{\{p_{k}\},\{\tau_{k}\}, k\in\Phi_{\mathcal{I}}}~~ &  \frac{ \sum_{ k\in\Phi_{\mathcal{I}}}\tau_kW\log_2\left(1+p_k\gamma_k\right)}
{\sum_{ k\in\Phi_{\mathcal{I}}}(\frac{p_k}{\varsigma}\tau_k+p_{c}\tau_k)}  \nonumber \\
\text{s.t.}  ~~& \text{C2:}~~ \frac{p_{k}}{\varsigma}\tau_{k}+p_{c}\tau_k \leq Q_k, ~~  k\in\Phi_{\mathcal{I}},  \nonumber\\
~~& \text{C3:}~~ \sum_{k\in\Phi_{\mathcal{I}}}\tau_k\leq T_{\rm{\max}}, \nonumber \\
~~& \text{C4:}~~ \tau_k\geq  0, ~k\in \Phi_{\mathcal{I}},    \nonumber\\
~~& \text{C5:}~~ p_k\geq  0, ~k\in \Phi_{\mathcal{I}},
\end{align}
where $EE^*_{\rm{IELCN}}$ is used to denote the maximum system EE of the IELCN.}

{If $EE^*_{\rm{PWPCN}} \geq EE^*_{\rm{IELCN}}$, then $EE^*=EE^*_{\rm{PWPCN}}$; otherwise $EE^*=EE^*_{\rm{IELCN}}$, i.e., either problem (\ref{eq132}) or problem (\ref{eq133}) provides the optimal solution for problem (\ref{eq10}).}
 \end{theorem}

 \begin{proof}
 Please refer to  Appendix A.
 \end{proof}

 Theorem \ref{theorem0} reveals that the  EE maximization problem in WPCN with initial stored energy can be cast into the EE maximization in one of the  two simplified systems, i.e., PWPCN or IELCN. In the following, we study the EE  and characterize its properties for each of the systems independently. {Note that for the special case that $EE^*_{\rm{PWPCN}} = EE^*_{\rm{IELCN}}$, without loss of generality, we assume that the system EE of problem (\ref{eq10}) is achieved by PWPCN in order to preserve the initial energy of users belonging to $\Phi_{\mathcal{I}}$. }  In the following, we study the EE as well
as characterizing the properties of each system independently.

\subsection{Properties of Energy-Efficient PWPCN}
The following lemma characterizes the operation of the power station for energy-efficient transmission.
\begin{lemma}\label{theorem1}
In energy-efficient PWPCN,  the power station always transmits with its maximum allowed power, i.e., $P_0 =P_{\mathop{\max}}$, for DL WET.
\end{lemma}
\begin{proof}
Please refer to  Appendix B.
\end{proof}
\begin{remark}
{This lemma seems contradictory to intuition at first. In conventional non-WPCN systems, since only the transmit power is optimized, the EE is generally first increasing and then decreasing with the transmit power when the circuit power is taken into account \cite{derrick_wipt,wu2012green,li2014energy,xiaomingchen_13,miao2013energy1,qing1,qing15_scheduling}.  Yet, in
 PWPCN, where the transmission  time can also be optimized, letting the power station transmit with the maximum allowed power reduces the time  needed for WET  in the DL, and thereby reduces the energy consumed by the circuits of the power station. Moreover, it also gives the users more time to improve the system throughput for WIT in the UL.} 
\end{remark}

The following lemma characterizes the time utilization for energy-efficient transmission.
\begin{lemma}\label{theorem2}
In energy-efficient PWPCN,
 the maximum system EE can always be  achieved by using up all the available transmission time, i.e.,
$\tau_0+\sum_{k\in \Phi_{\mathcal{P}}}\tau_k=T_{\max}.$
\end{lemma}
\begin{proof}
Please refer to  Appendix C.
\end{proof}

\begin{remark}{Lemma \ref{theorem2} indicates that, in PWPCN, using up the entire available transmission time is optimal.}
In fact, if the total available  time is not completely used up, increasing the  time for both DL WET and UL WIT by the same factor maintains the system EE at least at the same level, while improving the system throughput.
\end{remark}

Next, we study how the wireless powered users are scheduled for utilizing their harvested energy for energy-efficient transmission. 
\begin{lemma}\label{theorem3}
In energy-efficient PWPCN, the following scheduling strategy is optimal:\\
1)  If $EE^*_{\rm{PWPCN}}<ee^{\star}_k$, $\forall\, k\in \Phi_{\mathcal{P}}$,  then  user $k$ is scheduled, i.e., $\tau^*_k>0$,  and  it will use up all of its energy, i.e., $\tau^*_k(\frac{p^*_k}{\varsigma}+p_c)=\eta P_{\max}\tau^*_0h_k$.\\
2)  If $EE^*_{\rm{PWPCN}}=ee^{\star}_k$, $\forall\, k\in \Phi_{\mathcal{P}}$, scheduling user $k$ or not does not affect the maximum system EE, i.e., $0\leq \tau^*_k(\frac{p^*_k}{\varsigma}+p_c)\leq  \eta P_{\rm{\max}}\tau^*_0h_k$.\\
3)  If $EE^*_{\rm{PWPCN}}>ee^{\star}_k$, $\forall\, k\in \Phi_{\mathcal{P}}$, then  user $k$ is not scheduled, i.e., $\tau^*_k=0$, and it preserves all of its energy for the next transmission slot.
\end{lemma}
\begin{proof}
Please refer to  Appendix D.
\end{proof}

 Lemma \ref{theorem3}  reveals an important property related  to user scheduling and the corresponding energy utilization:   users that are scheduled should have a better or at least the same EE as the overall system, and for users with a strictly better EE, utilizing all of their energy always benefits the system EE.
\begin{remark}
In \cite{ju14_throughput}, the authors focus on the throughput maximization problem for PWPCN.  For that problem, the optimal transmission time of each user increases  linearly  with the equivalent channel gain. In other words, all users are scheduled no matter how severely their channel conditions are degraded. However, for EE oriented systems, it is not cost effective to schedule all users, especially if their channels are weak,  since each user introduces additional circuit power consumption.
\end{remark}

In Lemmas \ref{theorem1}, \ref{theorem2}, and \ref{theorem3}, we have revealed several basic properties of EE optimal PWPCN.   In the following, we derive an  expression for the maximum  EE and also the optimal solution based on the above  properties.
\begin{theorem}\label{theorem4}
The optimal system EE of PWPCN can be expressed as
\begin{align}\label{eq15}
EE_{\rm{PWPCN}}^*=\frac{\sum_{k\in S^*} ee^{\star}_kh_k}{\frac{1}{\eta} \left(\frac{P_{c}}{ P_{\mathop{\max}}}+\frac{1}{\xi}-\sum_{k=1}^{K}\eta h_k\right)+ \sum_{k\in S^*}h_k},
\end{align}
where $S^*\subseteq \Phi_{\mathcal{P}}$ is the optimal scheduled user set.  The optimal power and time allocation can be expressed as
\begin{align}
p^*_{k}&=\left[\frac{ W\varsigma }{ ee^{\star}_{k}\ln2}-\frac{1}{\gamma_{k}}\right]^+, \label{eq15_2}\\
\tau^*_0&\in\left(0, \frac{T_{\max}}{1+\eta P_{\max}\sum_{k\in S^*}\frac{ h_k ee^{\star}_k}{W\log_2(\frac{W\varsigma\gamma_k}{ee^{\star}_k\ln2})}}\right], \label{eq15_1}\\
\tau^*_k &= \eta P_{\max} \tau_0\frac{ h_k ee^{\star}_k}{W\log_2(\frac{W\varsigma\gamma_k}{ee^{\star}_k\ln2})}.  \label{eq15_3}
\end{align}
\end{theorem}
\begin{proof}
Please refer to  Appendix E.
\end{proof}

Theorem \ref{theorem4} provides a simple expression for the system EE in terms of the user EE and other system parameters. In (\ref{eq15}), since $P_{\rm{\max}}$ and  $P_c$ are  the maximum allowed transmit power and the circuit power, respectively,  their ratio $\frac{P_{c}}{P_{\rm{\max}}}$ can be interpreted as the inefficiency of the power station. The term $\frac{1}{\xi}-\sum_{k=1}^{K}\eta h_k$ represents the energy loss per unit transmit energy due to the  wireless channels, non-ideal energy harvesting devices, and a non-ideal PA at the power station.

Note that $\frac{1}{\eta} \left(\frac{P_{c}}{ P_{\mathop{\max}}}+\frac{1}{\xi}-\sum_{k=1}^{K}\eta h_k\right)$ involves only fixed system parameters and  is therefore a constant. This  means that  once $S^*$ is determined, the optimal solution can be obtained from (\ref{eq15}). Therefore, the problem is simplified to finding the optimal user set $S^*$. {In \cite{qing1}, we have proposed a linear-complexity algorithm for solving a scheduling problem with a similar structure as (\ref{eq15}).}
The details of this algorithm are omitted here and we refer the readers to \cite{qing1} for more information.

 Another interesting observation for PWPCN  is the relationship between the number of scheduled users and the physical system parameters, which has been summarized in the following corollary.
\begin{corollary}\label{corollary1}
{1) For energy-efficient PWPCN, the number of scheduled user increases with the ratio $\frac{P_c}{P_{\mathop{\max}}}$;
2) For energy-efficient PWPCN, the number of scheduled user decreases with the energy conversion efficiency $\eta$.}
\end{corollary}
\begin{proof}
{Due to space limitation, we only provide a sketch of the proof here.  From Lemma \ref{theorem3}, we know that the condition for scheduling user $k$ is $EE^*_{\rm{PWPCN}}\leq ee^{\star}_k$.  Since a larger  $\frac{P_c}{P_{\mathop{\max}}}$ or a lower  $\eta$ leads to a lower system EE, $EE^*_{\rm{PWPCN}}$, i.e., more users satisfy the scheduling condition, more users are scheduled.}
\end{proof}

{Corollary \ref{corollary1} generally reveals the relationship between the number of scheduled users and the physical system parameters of the power station ($P_c$, $P_{\mathop{\max}}$) and/or  user terminals ($\eta$) in the energy efficient PWPCN.} 

{In the next subsection, we investigate the EE of IELCN and characterize its properties.}
\subsection{Properties of Energy-Efficient IELCN}
\begin{theorem}\label{theorem01}
Problem (\ref{eq133}) is equivalent to the following optimization problem
  \begin{align}\label{eq134}
\mathop {\max }\limits_{ k\in\Phi_{\mathcal{I}}}~~ \mathop {\max }\limits_{p_{k},\tau_{k}}~~ &  \frac{ \tau_kW\log_2\left(1+p_k\gamma_k\right)}
{\frac{p_k}{\varsigma}\tau_k+p_{c}\tau_k}  \nonumber \\
\text{s.t.}   ~~~
& \frac{p_{k}}{\varsigma}\tau_{k}+p_{c}\tau_k \leq Q_k, ~~  k\in\Phi_{\mathcal{I}},  \nonumber\\
& \tau_k\leq T_{\rm{\max}}, ~~  k\in\Phi_{\mathcal{I}}, \nonumber\\
& \text{C4},~\text{C5},
\end{align}
and the corresponding optimal solution is given by
\begin{align}
p^*_k= &\left\{
\begin{array}{lcl}
p^{\star}_k,&& \text{if} ~ k=\arg\mathop {\max }\limits_{i\in\Phi_{\mathcal{I}}} ~ee^{\star}_{i},\\
0,&& \text{otherwise},~ \forall i,
\end{array}\right.  \label{eq135}\\
\tau^*_k& \left\{
\begin{array}{lcl}
\in\left(0, \max (\frac{Q_k}{\frac{p^*_k}{\varsigma}+p_c},T_{\max})\right],&&\text{if} ~ k=\arg\mathop {\max }\limits_{i\in\Phi_{\mathcal{I}}} ~ee^{\star}_{i},\\
=0,&& \text{otherwise},~ \forall i.
\end{array}\right.\label{eq136}
\end{align}
\end{theorem}
\begin{proof}
Please refer to  Appendix F.
\end{proof}

Theorem \ref{theorem01} indicates that the optimal transmission strategy for EE maximization in IELCN is to schedule only the user with the highest user EE. Thus, based on Theorem \ref{theorem01},  $EE^*_{\rm{IELCN}}$ can be easily obtained with the user EE introduced in Section II-C.

\textcolor{blue}{In summary,  we have obtained the optimal solutions of problems (11) and (12) in Section III-B and Section III-C, respectively. Thus,  as shown in Theorem \ref{theorem0},  the optimal solution of problem (10) is achieved by the one which results in larger system EE.}

\section{Energy-Efficient Resource Allocation  for WPCN with a QoS Constraint}
Since practical systems may have to fulfill certain QoS requirements, in this section, we  investigate energy-efficient time allocation and power control for WPCN guaranteeing a minimum system throughput. In this case, the EE maximization problem can be formulated as
\begin{align}\label{eq16}
\mathop {\max }\limits_{P_{0}, \tau_0,\{p_{k}\},\{\tau_{k}\}}~~ &\frac{\sum_{k=1}^{K} \tau_kW\log_2\left(1+p_k\gamma_k\right)}
{{P_{0}}\tau_0(\frac{1}{\xi}-\sum_{k=1}^{K}\eta h_k )+P_c\tau_0 + \sum_{k=1}^{K}(\frac{p_k}{\varsigma}\tau_k+p_{c}\tau_k)} \nonumber \\
\text{s.t.} ~~~~~& \text{C1, C2, C3, C4, C5},  \nonumber \\
& \text{C6:}~~\sum_{k=1}^{K}\tau_kW\log_{2}\left(1+p_{k}\gamma_{k}\right) \geq R_{\mathop{\min}},
\end{align}
where $R_{\min}$ denotes the minimum required system throughput and all other parameters and constraints are identical to those in (\ref{eq10}).  We note that different priorities and fairness among the users could be realized by adopting the weighted sum rate instead of the system throughput. However, since the weights are constants and do not affect the algorithm design, without loss of generality, we assume all users are equally weighted  in this paper \cite{derrick_wipt}.

{\subsection{ Feasibility of Problem (\ref{eq16})} }
{Before proceeding to solve problem (\ref{eq16}), we first investigate the feasibility condition for a given QoS requirement, $R_{\min}$. The following theorem provides the necessary and sufficient condition for the feasibility of problem (\ref{eq16}).
\begin{theorem}\label{feasibility}
Problem (\ref{eq16}) is feasible if  $R^*\geq R_{\mathop{\min}}$,  where $R^*$ is the maximum objective value  of the following concave optimization problem
\begin{align}\label{eq_feasi}
R^* \triangleq \mathop {\max }\limits_{\tau_0,\{\tau_{k}\}} &~~~ \sum_{k=1}^{K}\tau_kW\log_{2}\left(1+
\frac{\tau_0P_{\mathop{\max}}\eta h_{k}+Q_k}{\tau_k}\varsigma\gamma_{k}-p_{\rm{c}}\varsigma\gamma_k\right)   \nonumber\\
\text{s.t.}~ &~~~ \tau_{0}+\sum_{k=1}^{K}\tau_k= T_{\max}, \nonumber\\
&~~~ \tau_{0}\geq0, ~  \tau_k\geq  0, ~\forall\, k.
\end{align}
\end{theorem}
\begin{proof}
Due to the space limitation, we only provide a sketch of the proof. It can be shown that the maximum throughput of problem (\ref{eq16}) is achieved when C1-C3 are all satisfied with equality, which leads to problem (\ref{eq_feasi}). If the energy of some user is not used up, the system throughput can always be improved by increasing its transmit power while keeping its transmission time unchanged, thus C2 holds with strict equality. Similar considerations can also be made for C1 and C3, respectively.   The objective function in (\ref{eq_feasi}) is concave and all constraints are affine, thus problem (\ref{eq_feasi}) is a standard concave optimization problem.
\end{proof}}

{
In fact, problem (\ref{eq_feasi}) falls into the category of throughput maximization problems in WPCN and can be solved by standard optimization techniques, such as the interior point method \cite{Boyd}. The feasibility of problem (\ref{eq16}) can thereby be verified based on Theorem \ref{feasibility}.
}
{If it is infeasible, $R_{\min}$ can be decreased and/or $T_{\max}$ ($P_{\max}$) can be increased until the problem becomes feasible.} {In the following, we assume that problem (\ref{eq16}) is feasible.}

\subsection{Transformation of the Objective Function}
It is  intuitive that when   $R_{\mathop{\min}}$ is sufficiently large, both power transfer and the initial energy are needed to meet the system throughput requirement. Thus, problem (\ref{eq16}) cannot be simply cast into PWPCN or IELCN.
Moreover,  problem (\ref{eq16}) is neither convex nor quasi-convex due to the fractional form of the objective function and the non-convexity of
inequality constraints C2 and C6. 
Next, we study the transmit power of the power station.
\begin{theorem}\label{theorem6}
For problem (\ref{eq16}), the maximum system EE can always be achieved for $P^*_0=P_{\max}$.
\end{theorem}
\begin{proof}
As  the power transfer may not be activated due to the initial energy of the users,
  we discuss the following two cases. First, if the power transfer is activated for the optimal solution, i.e., $\tau^*_0>0$, then we can show that $P^*_0=P_{\max}$ following a similar proof as for Lemma \ref{theorem1}.   Second, if $\tau^*_0=0$  holds, then  the value of the  power station's  transmit power  $P^*_0$ does not affect the maximum system EE, and thus $P_0=P_{\max}$ is also an optimal solution.
\end{proof}

 It is worth noting that Lemma 1 is in fact a special case of Theorem \ref{theorem6}. Considering Theorem \ref{theorem6},  we only have to optimize $\tau_0$, $\{p_k\}$, and $\{\tau_k\}$, $\forall\, k$,  for solving problem (\ref{eq16}).
According to nonlinear fractional programming theory \cite{dinkelbach1967nonlinear}, for a problem of the form,
\begin{equation}\label{eq11c}
q^*=\max \limits_{\tau_{0},\{p_{k}\},\{\tau_{k}\}\in \mathcal{F}} \frac{B_{\rm{\rm{tot}}}(p_k,\tau_k)}{E_{\rm{\rm{tot}}}(\tau_0,p_k,\tau_k)},
\end{equation}
where $\mathcal{F}$ is the feasible set,
there exists an equivalent problem in  subtractive form, which satisfies
\begin{equation}\label{eq12c}
T(q^*)=\max \limits_{\tau_{0},\{p_{k}\},\{\tau_{k}\}\in \mathcal{F}}\{B_{\rm{\rm{tot}}}(p_k,\tau_k)-q^*E_{\rm{\rm{tot}}}(\tau_0,p_k,\tau_k)\}=0.
\end{equation}
The equivalence of (\ref{eq11c}) and (\ref{eq12c}) can be easily verified at the optimal point $(\tau^*_0,p^*_k,\tau^*_k)$ with the corresponding maximum value $q^*$ which  is  the optimal system EE to be determined. Dinkelbach
 provides an iterative  method in \cite{dinkelbach1967nonlinear} to obtain  $q^*$. In each iteration, a subtractive-form maximization problem (\ref{eq12c}) is solved for a given $q$. The value of $q$ is updated and problem (\ref{eq12c}) is solved again in the next iteration until convergence is achieved.
  By applying this transformation to (\ref{eq11c}), we obtain the following problem for a given $q$ in each iteration
\begin{align}\label{eq13c}
\max \limits_{\tau_{0},\{p_{k}\},\{\tau_{k}\}}~~~ &\sum_{k=1}^{K}\tau_kW\log_{2}\left(1+p_{k}\gamma_{k}\right)-q\left(P_{\max}\tau_0\left(\frac{1}{\xi}-\sum_{k=1}^K\eta h_k\right) +P_{c}\tau_0\textcolor{white}{\frac{1^1}{()}} \right.\nonumber\\
&\left. +\sum_{k=1}^{K}\left(\frac{p_k}{\varsigma}\tau_k  + p_{c}\tau_k\right)\right) \nonumber\\
\text{s.t.} ~~~~~& \text{C2, C3, C4, C5, C6}.
\end{align}

 Although problem (\ref{eq13c}) is  more tractable than the original problem (\ref{eq16}), it is still a non-convex optimization problem since it involves products of optimization variables. Hence,  we further introduce a set of auxiliary variables, i.e., $E_k=p_k\tau_k$, for $\forall \,k$, which can be interpreted as the actual energy consumed by user  $k$. Replacing  $p_k$ with $\frac{E_k}{\tau_k}$,  problem (\ref{eq13c}) can be written as
\begin{align}\label{eq14c}
\max \limits_{\tau_{0},\{E_{k}\},\{\tau_{k}\}}~~ &\sum_{k=1}^{K}\tau_kW\log_{2}\left(1+\frac{E_{k}}{\tau_k}\gamma_{k}\right)- q\left(P_{\max}\tau_0\left(\frac{1}{\xi}-\sum_{k=1}^K\eta h_k\right) +P_{c}\tau_0\textcolor{white}{\frac{1^1}{()}} \right.\nonumber\\
&\left. +\sum_{k=1}^{K}\left( \frac{E_k}{\varsigma}+p_{c}\tau_k\right)\right) \nonumber\\
\text{s.t.}~~~~ &\text{C3},~~\text{C4},  ~~\text{C5:}~~E_k\geq  0, ~\forall\, k, \nonumber\\
&\text{C2:}~~\frac{E_k}{\varsigma}+p_{c}\tau_k \leq \eta P_{\max}\tau_0h_k+Q_k, ~\forall\, k,  \nonumber\\
&\text{C6:}~~\sum_{k=1}^{K}\tau_kW\log_{2}\left(1+\frac{E_{k}}{\tau_k}\gamma_{k}\right) \geq R_{\mathop{\min}}.
\end{align}

After this substitution, it is easy to show that  problem (\ref{eq14c}) is a standard convex optimization problem, which can be solved by standard convex optimization techniques, e.g., the  interior-point method \cite{Boyd}. However, this
method neither exploits the particular structure of the problem itself nor does it provide any useful insights into the solution.
Hence, in the following, we employ the KKT conditions to analyze problem (\ref{eq14c}), which results in an optimal and  efficient solution.

\subsection{Iterative Algorithm for Energy Efficiency Maximization}
The partial  Lagrangian function of  problem (\ref{eq14c}) can be written as
\begin{align}\label{eq161}
~\mathcal{L}(\tau_0,E_k,\tau_k, \bm{\mu}, \delta, \vartheta)
=\,&(1+\vartheta)\sum_{k=1}^{K}\tau_kW\log_{2}\left(1+\frac{E_{k}}{\tau_k}\gamma_{k}\right) +\delta\left(T_{\max}-\tau_0-\sum_{k=1}^{K}\tau_k\right) \nonumber\\
-\,&q \left(P_{\max}\tau_0\left(\frac{1}{\xi}-\sum_{k=1}^{K}\eta h_k\right)+P_{c}\tau_0+ \sum_{k=1}^{K}\left( \frac{E_k}{\varsigma}+p_{c}\tau_k\right)\right) -\vartheta R_{\mathop{\min}}\nonumber \\ +\,&\sum_{k=1}^{K}\mu_k\left(Q_k+ \eta P_{\max}\tau_0h_k-\frac{E_k}{\varsigma}-p_{c}\tau_k\right),
\end{align}
where  $\bm{\mu}$, $\delta$, and  $\vartheta$ are the non-negative Lagrange multipliers associated with  constraints C2, C3,  and C6, respectively.  The boundary constraints C4 and C5 are absorbed into the optimal solution in the following. Then, the optimal solution can be obtained from the following theorem.
 \begin{theorem}\label{theorem70}
Given $\bm{\mu}$, $\delta$, and $\vartheta$,  the maximizer of $\mathcal{L}(\tau_0,E_k,\tau_k, \bm{\mu}, \delta, \vartheta)$ is given by
\begin{align}
\tau^{*}_0& \left\{
\begin{array}{lcl}
\in[\, 0, T_{\max}),&\text{if} ~f_0(\bm{\mu})=0,\\
=0,&\text{if} ~f_0(\bm{\mu})<0,
\end{array}\right.\label{solu_1}\\
E^*_k&=\tau^{*}_kp_k, \forall\, k,\label{solu_2}\\
\tau^{*}_k& \left\{
\begin{array}{lclcl}
=\frac{\eta P_{\max}\tau^{*}_0h_k+Q_k}{\frac{p^*_k}{\varsigma}+p_c}, &\text{if} ~\gamma_k>x^*,\\
 \in \left[0,\frac{\eta P_{\max}\tau^{*}_0h_k+Q_k}{\frac{p^*_k}{\varsigma}+p_c}\right],&\text{if} ~\gamma_k=x^*,\\
=0,&\text{if} ~\gamma_k<x^*,
\end{array}\right. \label{solu_3}
\end{align}
where $p_k$ and $f_0(\bm{\mu})$ are  given by
\begin{align}
p^*_k&=\left[\frac{W(1+\vartheta)\varsigma}{(q+\mu_k)\ln2}-\frac{1}{\gamma_k}\right]^+, \forall \, k,\label{solu_4}\\
f_0(\bm{\mu}) &= P_{\max}\left(\sum_{k=1}^{K}\mu_kh_k-q\left(1-\sum_{k=1}^{K}\eta h_k\right)\right)-qP_c-\delta.\label{solu_5}
\end{align}
In (\ref{solu_3}), $x^*$ denotes the solution of
\begin{align}\label{solu_50}
aq(\ln2)\log_2(ax)\frac{1}{\varsigma}+\frac{q}{x}-q(a+p_c)-\delta=0,
\end{align}
where $a\triangleq \frac{W(1+\vartheta)\varsigma}{q\ln2}$.
 \end{theorem}
 \begin{proof}
 Please refer to Appendix G.
\end{proof}
\begin{algorithm}[!hbpt]
 \caption{Energy-Efficient Transmission Algorithm for WPCN}
  \begin{algorithmic}[1]
\item Initialize $q=0$ and the maximum tolerance $\epsilon$;
\item  {\bf{Repeat}}  \\
           ~~~~Initialize $\hat{\vartheta}$ and $\hat{\delta}$;\\
           ~~~~Set Lagrange multipliers  $\vartheta_{\mathop{\max}}=\hat{\vartheta}$,   $\vartheta_{\mathop{\min}}=0$,  $\delta_{\mathop{\max}}=\hat{\delta}$, and $\delta_{\mathop{\min}}=0$; \\
           ~~~~{\bf{While}} $\vartheta_{\mathop{\max}}-\vartheta_{\mathop{\min}} \geq  \epsilon$
\item
          ~~~~~~~~$\vartheta=\frac{1}{2}(\vartheta_{\mathop{\max}}+\vartheta_{\mathop{\min}})$;\\
          ~~~~~~~ ~~~~{\bf{While}} $\delta_{\mathop{\max}}-\delta_{\mathop{\min}} \geq  \epsilon$\\
          ~~~~~~~~~~~~~~~~$\delta=\frac{1}{2}(\delta_{\mathop{\max}}+\delta_{\mathop{\min}})$;\\
          ~~~~~~~~~~~~~~~~Compute $x^*$  from  (\ref{solu_50}) for given $q$, $\vartheta$, and $\delta$;  \\
          ~~~~~~~~~~~~~~~~Compute $\mu_k$ from (\ref{apdx_eq24}) and  (\ref{apdx_eq25}) with $\gamma_k>x^*$; otherwise, $\mu_k=0$;\\
          ~~~~~~~~~~~~~~~~Obtain  $p_{k}$ for each user from  (\ref{solu_4}); \\ 
          ~~~~~~~~~~~~~~~~Obtain $\tau_0$ and $\tau_k$ from (\ref{solu_1}) and (\ref{solu_3}), respectively;\\
          ~~~~~~~~~~~~~~~~{\bf{If}} there exist $\tau_0$ and $\tau_k, \forall \, k$, satisfying (\ref{solu_7}), then,  {\bf{break}};     \\
          ~~~~~~~~~~~~~~~~{\bf{elseif}} $  \tau_{0}+\sum_{k=1}^{K}\tau_k > T_{\max} $       \\
                        ~~~~~  ~~~~~~~~~~~~~$\delta_{\mathop{\min}}=\delta$;  ~~{\bf{else}}
                      ~~ $\delta_{\mathop{\max}}=\delta$; \\
                                  ~~~~~~~~~~~~~~~~{\bf{end}}     \\
                                            ~~~~~~~ ~~~~{\bf{end while}}  \\
                                            ~~~~~~~ ~~~~{\bf{If}} power allocation variables $p_k,\, \forall\, k $, satisfying (\ref{solu_8}), then,  {\bf{break}};     \\
                                            ~~~~~~~ ~~~~{\bf{elseif}}  $ \sum_{k=1}^{K}\tau_kW\log_{2}\left(1+p_kg_{k}\right) < R_{\mathop{\min}} $   \\
          ~~~~~~~~~~~~~~~$\vartheta_{\mathop{\min}}=\vartheta$;  ~~{\bf{else}}
                      ~~ $\vartheta_{\mathop{\max}}=\vartheta$; \\
                                            ~~~~~~~ ~~~~{\bf{end}}     \\
           ~~~~{\bf{end while}}  \\
           \ \ \ \ Update  $q=\frac{B_{\rm{\rm{tot}}}({p_k},{\tau_k})}{E_{\rm{\rm{tot}}}(\tau_0,  {p_k},{\tau_k})}$;    \\
           {\bf{until}}       $ T(q^*)< \epsilon$
\end{algorithmic}
\end{algorithm}

{By exploiting Theorem \ref{theorem70}, the optimal solution of (\ref{eq14c}) can be obtained with Algorithm 1 given on the next page. In Algorithm 1, we first initialize the Lagrange multipliers $\vartheta$ and $\delta$. Line 9 calculates $x^*$ from (\ref{solu_50}), where $x^*$ is the threshold to determine whether a user is scheduled or not. It is interesting to note that since the parameters $a$, $q$, $\varsigma$, $p_c$, and $\delta$ in (\ref{solu_50}) are independent of the user index $k$, the threshold $x^*$ is thereby identical for all users.
 Then,  based on (\ref{solu_3}), we determine the users that should be scheduled by comparing $x^*$ with $\gamma_k$.  Thus, for an unscheduled user $k$, its corresponding $\mu_k$ is zero since constraint C2 is met with strict inequality. In contrast, for a scheduled user $k$ with $\gamma_k$, line 10 calculates its corresponding $\mu_k$ by setting $f(\mu_k, \gamma_k)=0$ in (\ref{apdx_eq25}), where $f(\mu_k, \gamma_k)$ is given by (\ref{apdx_eq24}).
  With given $\vartheta$, $\delta$, and $\mu_k$, the power allocation variable $p_k$ can be immediately computed from (\ref{solu_4})  in line 11. Then, from (\ref{solu_1}) and (\ref{solu_3}), the region with respect to $\tau_0$ and $\tau_k$ is easily obtained as in line 12. Since it has been shown in (\ref{apdx_eq20}) and (\ref{apdx_eq24}) that the Lagrangian function $\mathcal{L}$ is a linear function with respect to $\tau_0$ and $\tau_k$, the optimal solution that maximizes $\mathcal{L}$ can always be found at the vertices of the region created by $\tau_0$ and $\tau_k$.  }
  It is worth noting that in the case that all users have sufficient energy, it follows that $\mu_k=0$ due to the complementary slackness condition (\ref{apdx_eq200}). Then, from (\ref{solu_5}), this leads to $f_0(\bm{\mu})<0$  which implies that activating the power transfer is not beneficial for achieving the highest system EE.
Otherwise, $\vartheta$ and $\delta$ are updated iteratively until they converge. 

{ The computational complexity of Algorithm 1 can be analyzed as follows. The complexity of lines  8-11 in Algorithm 1 is linear in the number of users, $K$. Furthermore, the complexity of the Dinkelbach method  \cite{ng2012energy3} for updating  $q$ and  the bisection method \cite{Boyd} for updating $\vartheta$ and $\delta$ are both independent of $K$. Therefore, the total complexity of the proposed algorithm is $\mathcal{O}(K)$.}

In the following, we reveal some properties of energy-efficient WPCN with a throughput constraint.
\begin{corollary}\label{theorem7}
If the total available transmission time is not used up, i.e., $\tau_0+\sum_{k=1}^K\tau_k<T_{\max}$, then each scheduled user $k$ transmits with the power that achieves the maximum user EE, i.e., $p_{k}=p^{\star}_{k}$ in (\ref{eq11}). In contrast, if the total available transmission time is used up, then the optimal transmit power of each scheduled user $k$ satisfies $p_{k} \geq p^{\star}_k$.
\end{corollary}
\begin{proof}
Please refer to Appendix H.
\end{proof}

Corollary \ref{theorem7}  reveals that as long as the total available transmission time is sufficiently long,  letting each user independently maximize its own maximum EE is the most energy-efficient power control strategy for the whole system, which also coincides with the conclusion in Theorem \ref{theorem4} for best-effort PWPCN. On the other hand, if the available transmission time is not sufficient, users can only meet the required system throughput by increasing their transmit power at the expense of sacrificing user EE and also system EE.
{Furthermore, users that are not scheduled in the problem without $R_{\rm{\min}}$, i.e., problem (\ref{eq10}), may have to be scheduled in order to meet $R_{\rm{\min}}$, although scheduling them is detrimental to the system EE. Thus, it is likely that some of these users only consume just enough of their energy to satisfy $R_{\min}$.}
 The following corollary sheds some light on how an energy-efficient WPCN meets the QoS requirement.
\begin{corollary}\label{theorem8}
If WET is used, i.e., $\tau_0>0$,  and a scheduled user $m$ does not use up all of its available energy, then the transmit powers of all scheduled users remain constant until user $m$'s energy is used up. Moreover, as the required system throughput increases, the energy transfer time $\tau_0$  and the  transmission  time $\tau_k$ of any scheduled user $k\neq m$ decrease, respectively, while the transmission time $\tau_m$ of user $m$ increases.
\end{corollary}
\begin{proof}
Please refer to Appendix I.
\end{proof}

{Corollary \ref{theorem8} suggests that  if some scheduled user has a large amount of initial energy available,  it is preferable to utilize this energy  instead of prolonging  the DL WET time if the required throughput is high.  This is because DL WET not only causes circuit energy consumption but also reduces the time for UL WIT.}

\section{Numerical Results }

In this section, we present simulation results to validate our theoretical findings, and to demonstrate the system EE of WPCN.  Five users  are randomly and uniformly distributed on the right hand side of the power station with a reference distance of  2 meters and a maximum service distance of 15 meters. The information receiving  station is located 300 meters away from the power station. The system bandwidth is set as 20 kHz and the SNR gap is $\Gamma=0$ dB.  The path loss exponent is 2.8 and the thermal noise power is -110 dBm. The small scale fading for WET and WIT is Rician fading with Rician factor 7 dB and Rayleigh fading, respectively.  The circuit power consumptions at the power station and the user terminals are set to 500 mW and 5 mW \cite{li2014wireless}, respectively. The PA efficiencies of the power station and the user terminals, i.e., $\xi$ and $\varsigma$, are set to unity, without loss of generality. Unless specified otherwise, the remaining system parameters are  set to $\eta =0.9$,  $T_{\max}=1$s, and $P_{\max}=43$ dBm.  {In Figs. 3-6, best-effort communication WPCN are considered, whereas in Figs. 2 and 7, a minimum system throughput requirement is imposed.}
 \begin{figure}[!th]
\centering
\includegraphics[width=3in]{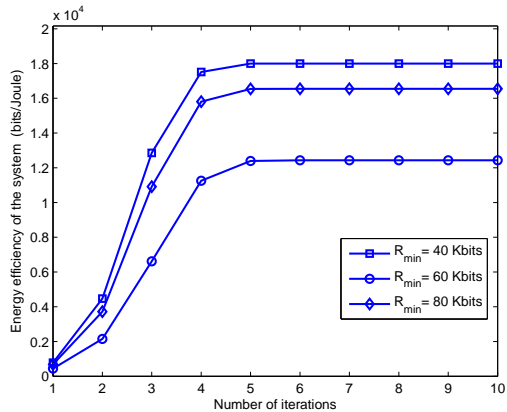}  
{\caption{ System EE versus the number of outer-layer iterations of the proposed algorithm for different minimum system requirements, $R_{\min}$. }\label{fig00}}
\end{figure}
{\subsection{ Convergence of Proposed Algorithm}
Fig. \ref{fig00} depicts the achieved system EE of the proposed Algorithm 1 versus the number of outer-layer iterations using the Dinkelbach method for different configurations. As can be observed, on average at most six iterations are needed to reach the optimal solution in the outer-layer optimization.  Since the time allocation and power control by the bisection method also results in a fast convergence in the inner-layer optimization \cite{Boyd}, the proposed algorithm is guaranteed to converge quickly.}
\subsection{ System EE of WPCN: PWPCN versus IELCN}
We provide a concrete example to illustrate Theorem \ref{theorem0} for best-effort communication. 
Specifically, we set {$\bm{Q}\triangleq [Q_1,Q_2,Q_3,Q_4,Q_5]=[0,  0, 1, 1, 1]$ (Joule),  $\bm{h}\triangleq [h_1,h_2,h_3,h_4,h_5]=[0.1, 0.1, 0.1, 0.1, 0.1]$, and $\bm{\gamma}\triangleq [\gamma_1,\gamma_2,\gamma_3,\gamma_4,\gamma_5]=[8,6,\gamma_3,0.3,0.2]$, respectively.} Note that only the  last three users have initial energy available. Therefore, from Theorem \ref{theorem01} for  IELCN, we know that only the third user is scheduled if $\gamma_3>0.3$, and its EE is independent of $P_{\max}$ and increasing with $\gamma_3$. However, from Theorem \ref{theorem4}, we know that the EE of PWPCN is increasing in $P_{\max}$. Therefore, we can vary $\gamma_3$ and $P_{\max}$ to observe  the system switching from IELCN to PWPCN in terms of system EE, which is  shown in Fig. \ref{fig30}. In the low transmit power regime, the system is in the IELCN mode,  but as  $P_{\max}$ increases, when the EE of PWPCN surpasses that of IELCN, the system switches to the PWPCN mode.
 \begin{figure}[!th]
\centering
\includegraphics[width=3in]{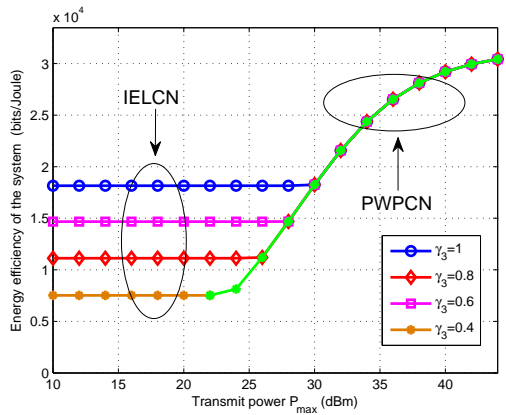}  
\caption{Illustration of the system switching from IELCN to PWPCN as $P_{\max}$ increases. The green curve corresponds to PWPCN and the horizontal portion of curves corresponds to IELCN.}\label{fig30}
\end{figure}

\subsection{System EE  versus Transmit Power of  Power Station and Path Loss Exponent of WET Channel}


We compare the EE of the following schemes: 1) EE Optimal: proposed approach; 2) Throughput Optimal: based on conventional throughput
maximization \cite{ju14_throughput};  3) Fixed Proportion: let each user consume a fixed proportion of its harvested energy, denoted as $\rho$, which can be adjusted to balance the energy consumed and stored.
 In Fig. \ref{fig40}, as $P_{\mathop{\max}}$  increases, we  observe that the performance of the EE Optimal scheme first sharply increases and then experience a moderate increase while the EE of the Throughput Optimal scheme first increases and then strictly decreases, which is due to its greedy use of  power. Moreover, for the Fixed Proportion schemes, as $\rho$ increases, the system EE also increases. However, even for $\rho=1$, the EE Optimal scheme still outperforms the Fixed Proportion scheme. The proposed scheme has a superior performance as it only schedules users which are beneficial for the system EE while the Fixed Proportion scheme imprudently schedules all users without any selection.

  In Fig. \ref{fig50}, the  system EE of all schemes decreases with increasing path loss exponent $\alpha$.  Moreover, the performance gap between the different schemes decreases as $\alpha$ increases.  A larger path loss exponent leads to more energy loss in signal propagation, which forces the energy-efficient designs to schedule more users and to utilize  more energy to increase the system throughput so as to improve the system EE. Hence, the proposed algorithm behaves similar to the Throughput Optimal scheme for very high path loss exponents.
%
%

\begin{figure*}[!t]
\begin{minipage}[t]{0.5\linewidth}
\centering
\includegraphics[width=3in, height=2.3in]{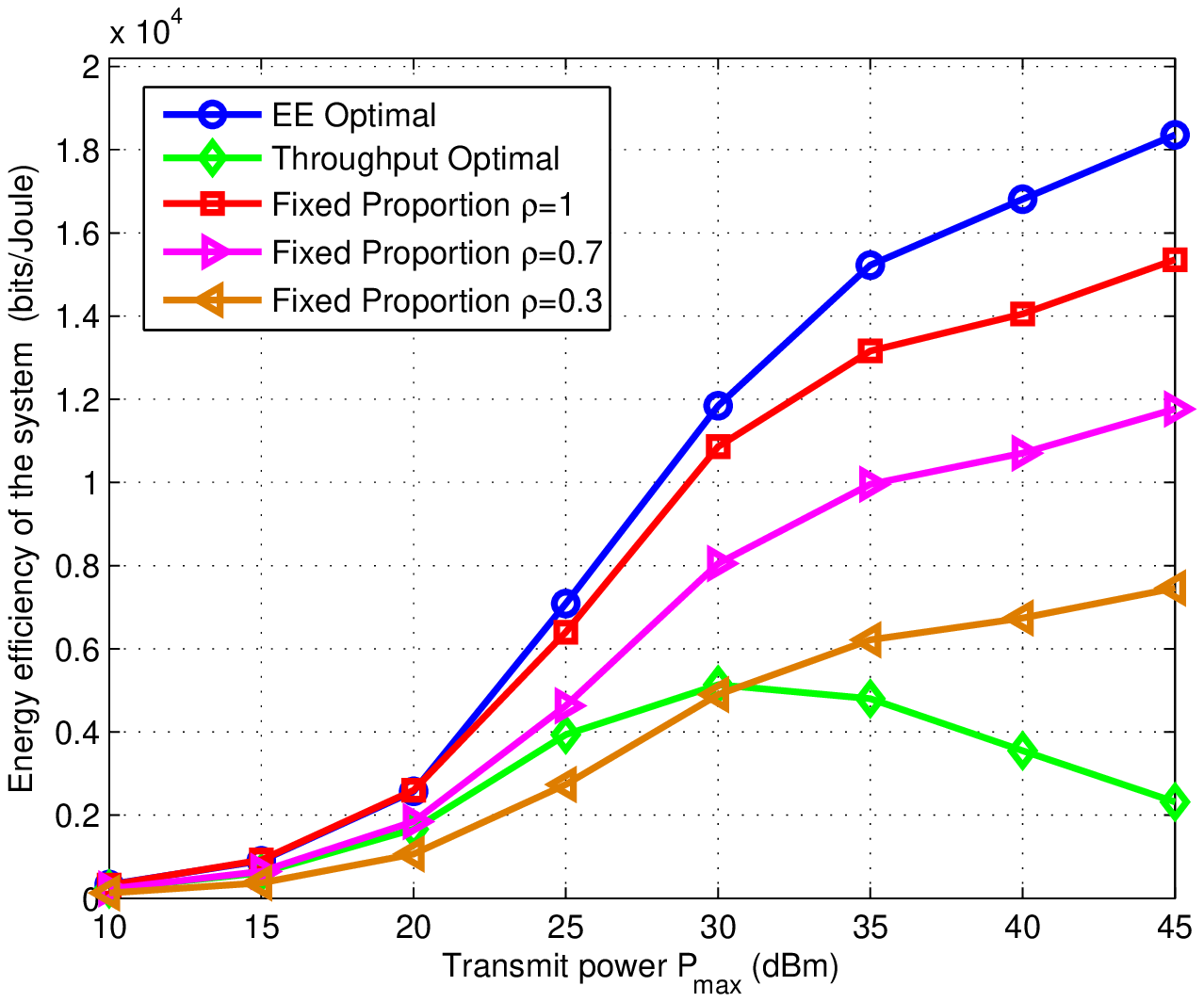}  
\caption{System EE versus the maximum transmit power.}\label{fig40}
\label{fig:side:a}
\end{minipage}%
\begin{minipage}[t]{0.5\linewidth}
\centering
\includegraphics[width=3in, height=2.3in]{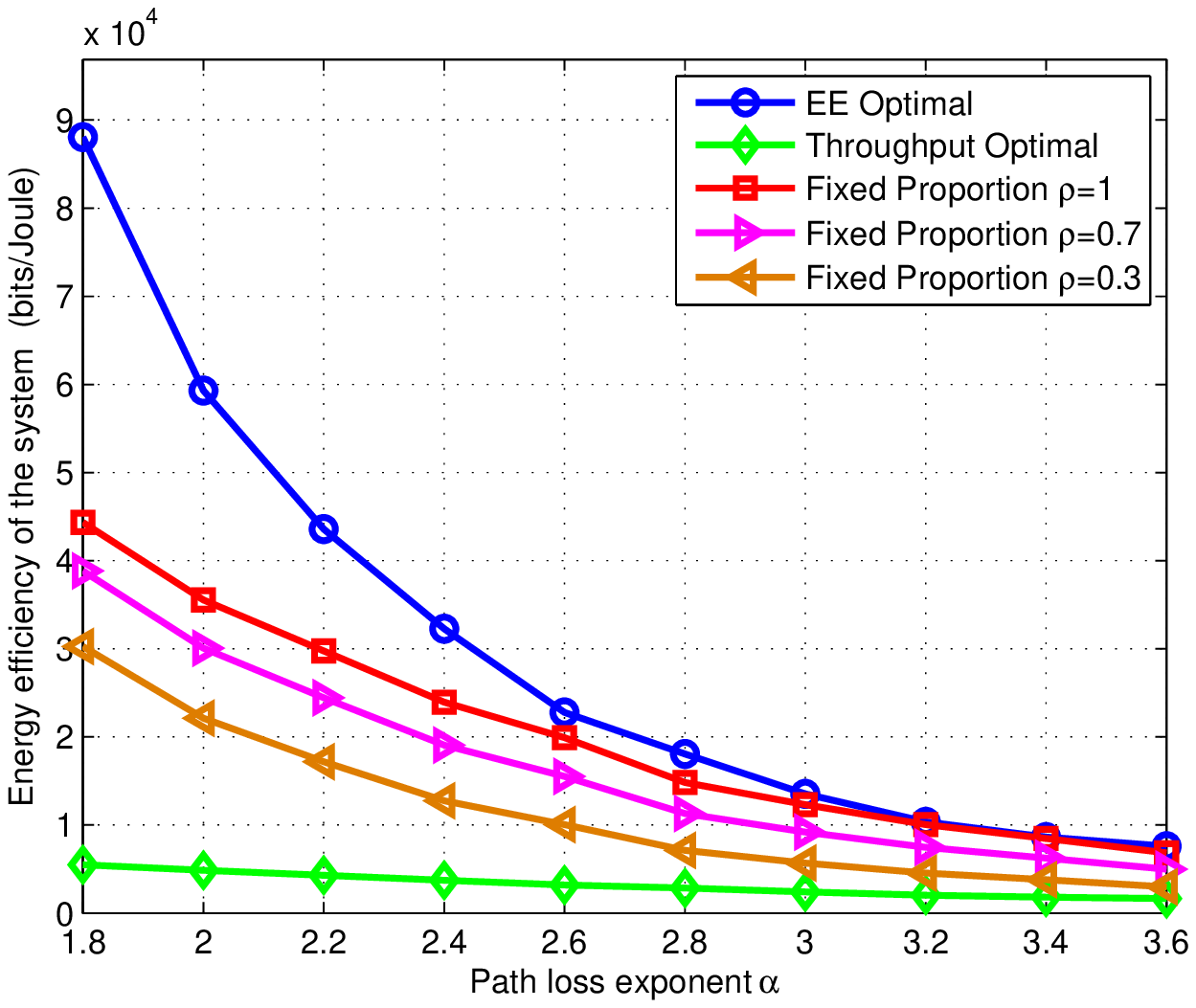}
\caption{System EE  versus the path loss exponent. }\label{fig50}
\end{minipage}%
\end{figure*}

\subsection{Number of Scheduled Users versus Energy Harvesting Efficiency}
In Fig. \ref{fig60}, we show the number of scheduled users versus the energy harvesting efficiency of the user terminal, $\eta$.  An interesting observation is that  the number of scheduled users is non-decreasing with increasing $\eta$.
This is because as the energy harvesting efficiency increases, the energy loss decreases which leads to a higher system EE. This further forces the system to be more conservative in scheduling users so as to maintain higher EE. Moreover, for a larger $P_c$, more users are scheduled.

\subsection{System EE  versus Minimum Throughput Requirement}

Fig. \ref{fig70} shows the system EE versus the minimum required system throughput, $R_{\min}$, for different numbers of user terminals.  We observe that as $R_{\min}$ increases, the system EE first remains constant and then gradually decreases, which is due to the fundamental trade-off between EE and spectral efficiency (SE). 
As expected, the EE increases with the number of users $K$. The reasons for this are twofold. First, for DL WET, if more users participate in energy harvesting, the energy loss due to signal propagation decreases. Second, for UL WIT, a larger number of users results in a higher multiuser diversity gain, which in turn leads to a higher system throughput.


%
%

\begin{figure*}[!t]
\begin{minipage}[t]{0.5\linewidth}
\centering
\includegraphics[width=3in, height=2.3in]{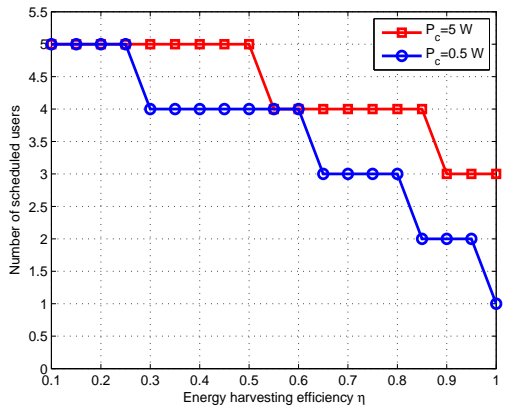}  
\caption{Number of scheduled users versus the efficiency $\eta$.}\label{fig60}
\label{fig:side:a}
\end{minipage}%
\begin{minipage}[t]{0.5\linewidth}
\centering
\includegraphics[width=3in, height=2.3in]{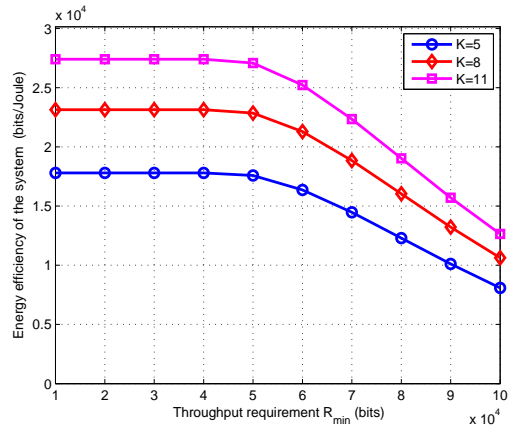}
\caption{System EE  versus the minimum required throughput.}\label{fig70}
\end{minipage}%
\end{figure*}
Another interesting observation is that for larger  $K$, the system EE decreases more rapidly than for smaller $K$. This is mainly because for larger $K$, more energy is harvested and thus the  energy loss in DL WET is relatively less dominant in the total energy consumption compared to the energy consumed for UL WIT. Therefore, for high throughput requirements, the energy consumption is more sensitive to changes in the throughput requirements for larger $K$, which leads to a faster decrease in the system EE.

\section{Conclusions}
{In this paper, we have  investigated the joint time allocation and power control of DL WET and  UL WIT to maximize the system EE of the WPCN. For the WPCN with best-effort communication, we have shown that the EE maximization problem is equivalent to the EE maximization in two different simplified systems, i.e., PWPCN and IELCN. For the PWPCN,  we have reduced the EE maximization problem to a multiuser scheduling problem where the number of scheduled users  increases with the circuit power but decreases with the energy conversion efficiency at the user side. On the other hand, for the IELCN, only the user with the highest user EE is scheduled.
 Furthermore, we have studied the EE maximization problem under a minimum required system throughput constraint and proposed an efficient algorithm for obtaining the optimal solution. In addition, we have shown that when the available transmission time is sufficiently long, the most energy-efficient strategy for the system is to let each user achieve its own maximum user EE. In contrast, if the transmission time is too short, the system EE has to be sacrificed to achieve the system throughput requirement.}

{
There are several interesting research directions that could be pursued based on the results in this paper:
1) 
While the throughput in UL WIT improves with the quality of the CSI, this comes at the expense of energy and time needed for CSI estimation which reduces the system EE.
Therefore, the design of the optimal CSI acquisition strategy for maximizing the system EE  is an interesting topic.
2) Beyond the system EE, maximizing the user EE may be desirable in practice, for example,  to extend the lifetime of some specific battery. Thus, the user EE tradeoff is worth studying so that different transmission strategies can be employed to strike the balance among EEs of different users.
3) Finally,  maximizing the system EE while guaranteeing minimum individual user throughputs is also an interesting problem.}

\appendices
\section*{Appendix A: Proof of Theorem \ref{theorem0}}\label{apdix0}
{We first introduce a lemma to facilitate our proof.
\begin{lemma}
Assume that  $a$, $b$, $c$, and $d$ are arbitrary positive numbers. Then, we have
$\frac{a+c}{b+d}\leq\max{\{\frac{a}{b},\frac{c}{d}\}}$ where ``=" holds if and only if $\frac{a}{b}=\frac{c}{d}$.
\end{lemma}
\begin{proof}
The proof is straightforward and thus omitted due to the space limitation.
\end{proof}}
{{Let $\mathcal{S}=\{P_0, \tau_0,\{{p}_{k}\}, \{{\tau}_k\}\}$ denote an \emph{arbitrary} solution of problem (\ref{eq10}) and its corresponding system EE is denoted as $EE$. Let $\mathcal{\widehat{S}}=\{\widehat{P}_0, \hat{\tau}_0,$  $\{\hat{p}_{k}\},\{\hat{\tau}_k\}\}$ and $\mathcal{\check{S}}=\{{\check{P}}_0, 0,\{\check{p}_{k}\},\{\check{\tau}_k\}\}$ denote the optimal solutions of problem (\ref{eq132}) or problem (\ref{eq133}), respectively.} The energy consumptions corresponding to $S$, $\widehat{S}$, and $\check{S}$ during DL WET are $E_{\rm{WET}}$, $\widehat{E}_{\rm{WET}}$, and 0, respectively.} The feasible sets of problems (\ref{eq10}),  (\ref{eq132}), and (\ref{eq133}) are denoted as $D$,  $D_{\mathcal{P}}$,  and $D_{\mathcal{I}}$, respectively, {and  $r_k(p_k)\triangleq W \log_2(1+{p_k\gamma_k})$.} {Note that if $\tau_k=0$ holds for $\forall \,  k\in\Phi_{\mathcal{P}}$ and  $\forall\, k\in\Phi_{\mathcal{I}}$, the system EE of WPCN is zero which is obviously not the maximum value of problem (\ref{eq10}).  Therefore, the maximum EE of problem (\ref{eq10}), $EE^*$, can only be achieved for one of the following three cases:\\
1) \{$\tau_0>0; \exists\,k\in\Phi_{\mathcal{P}}, \tau_k>0$;  $\forall \, k\in\Phi_{\mathcal{I}}, \tau_k=0$\}: In this case, as $\tau_0>0$ while $\forall \, k\in\Phi_{\mathcal{I}}, \tau_k=0$, the maximum EE of WPCN is achieved by PWPCN, i.e., problem (\ref{eq10}) simplifies to problem (\ref{eq132}) and $EE^*=\max{\{EE^*_{\rm{PWPCN}},0\}}=EE^*_{\rm{PWPCN}}.$\\
2) \{$\tau_0=0; \forall \,k\in\Phi_{\mathcal{P}}, \tau_k=0$;  $\exists\, k\in\Phi_{\mathcal{I}}, \tau_k>0$\}: In this case, as $\tau_0=0$ and  $\forall \,k\in\Phi_{\mathcal{P}}, \tau_k=0$, the maximum EE of WPCN is achieved by IELCN, i.e., problem (\ref{eq10}) simplifies to problem (\ref{eq133}) and $EE^*=\max{\{0,EE^*_{\rm{IELCN}}\}}=EE^*_{\rm{IELCN}}.$\\
 3) \{$\tau_0>0; \exists\, k\in\Phi_{\mathcal{P}}, \tau_k>0$;  $\exists\, k\in\Phi_{\mathcal{I}}, \tau_k>0$\}: In this case, by exploiting the fractional structures of (\ref{eq10})-(\ref{eq133}), we have the following inequalities
\begin{eqnarray}\label{apdx_eq01}
EE
&=& \frac{\sum_{k=1}^{K} \tau_kr_{k}(p_k)}
{E_{\rm{WET}}+ \sum_{k=1}^{K}\tau_k(\frac{p_k}{\varsigma}+p_{c})} \nonumber \\
&=&\frac{\sum_{k\in \Phi_{\mathcal{P}}} \tau_kr_k(p_k)+\sum_{ k\in\Phi_{\mathcal{I}}}\tau_kr_k(p_k)}
{E_{\rm{WET}} + \sum_{k\in\Phi_{\mathcal{P}}}\tau_k(\frac{p_k}{\varsigma}+p_{c})+ \sum_{k\in\Phi_{\mathcal{I}}}\tau_k(\frac{p_k}{\varsigma}+p_{c})} \nonumber \\
&\overset{a}\leq& \max\left\{\frac{\sum_{k\in \Phi_{\mathcal{P}}} \tau_kr_k(p_k)}{E_{\rm{WET}} + \sum_{k\in\Phi_{\mathcal{P}}}\tau_k(\frac{p_k}{\varsigma}+p_{c})},\frac{ \sum_{ k\in\Phi_{\mathcal{I}}}\tau_kr_k(p_k)}{\sum_{k\in\Phi_{\mathcal{I}}}\tau_k(\frac{p_k}{\varsigma}+p_{c})} \right\}    \nonumber \\
&\overset{b}\leq& \max\left\{\frac{\sum_{k\in \Phi_{\mathcal{P}}} \hat{\tau}_kr_k(\hat{p}_k)}{\widehat{E}_{\rm{WET}} + \sum_{k\in\Phi_{\mathcal{P}}}\hat{\tau}_k(\frac{\hat{p}_k}{\varsigma}+p_{c})},\frac{ \sum_{ k\in\Phi_{\mathcal{I}}}\check{\tau}_kr_k(\check{p}_k)}{\sum_{k\in\Phi_{\mathcal{I}}}\check{\tau}_k(\frac{\check{p}_k}{\varsigma}+p_{c})}\right\}    \nonumber \\
&=& \max\left\{EE^*_{\rm{PWPCN}}, EE^*_{\rm{IELCN}}\right\},
\end{eqnarray}
{where inequality ``$a$'' holds due to Lemma 4 and the strict equality ``='' represents the special case when the system EE of PWPCN is the same as that of IELCN.}
 Inequality ``$b$'' holds since  ${\widehat{S}}$ and $\check{S}$ are the optimal solutions corresponding to $EE^*_{\rm{PWPCN}}$ and  $EE^*_{\rm{IELCN}}$, respectively.  {Therefore,  in (35), if and only if the \emph{maximum} system EE of PWPCN is the same as the \emph{maximum} system EE of IELCN, the strict equality in ``$a$''  can hold together with the strict equality in ``$b$''. In this case, there exists a solution that satisfies $\{\tau_0>0; \exists\, k\in\Phi_{\mathcal{P}}, \tau_k>0$;  $\exists\, k\in\Phi_{\mathcal{I}}, \tau_k>0\}$ and achieves
the maximum system EE of WPCN. It thus follows that $EE^*=EE^*_{\rm{PWPCN}}=EE^*_{\rm{IELCN}}$ and without loss of generality, we assume that the maximum system EE for this case is achieved by PWPCN in order to preserve the initial energy of users belonging to $\Phi_{\mathcal{I}}$.
Otherwise, the strict equality in ``$a$''  can not hold together with the strict equality in ``$b$''.  This means that the system EE achieved by any solution that satisfies \{$\tau_0>0; \exists\, k\in\Phi_{\mathcal{P}}, \tau_k>0$;  $\exists\, k\in\Phi_{\mathcal{I}}, \tau_k>0$\} will be strictly smaller than the maximum EE of  either  PWPCN or IELCN, i.e., $EE^*=\max{\{EE^*_{\rm{PWPCN}},EE^*_{\rm{IELCN}}\}}$, which suggests that either PWPCN or IELCN is optimal.}}
{Next, we investigate under what conditions ``$b$'' holds with strict equality, i.e.,  $EE^*_{\rm{PWPCN}}$  and  $EE^*_{\rm{IELCN}}$ are achieved without violating the feasible domain of the original problem (\ref{eq10}). This leads to the following two cases:
\begin{itemize}
  \item For $k\in\Phi_{\mathcal{P}}$, it is easy to verify the equivalence between  $\{{P}_0, \tau_0,\{{p}_{k}\},\{{\tau}_k\}\}\in {D}$ and $\{{P}_0, \tau_0, \{{p}_{k}\},$ $\{{\tau}_k\}\}\in {D_\mathcal{P}}$. As $\{\widehat{P}_0, \hat{\tau}_0,\{\hat{p}_{k}\},\{\hat{\tau}_k\}\}$ maximizes $EE_{\rm{PWPCN}}$, ``$b$'' holds true for the first term inside the  bracket.
  \item For $k\in\Phi_{\mathcal{I}}$, the optimal solution, denoted as $\{{p}^*_{k}, {\tau}^*_k\}\in D$, implies that $\tau^*_k(\frac{p^*_k}{\varsigma}+p_c)\leq \eta P^*_0\tau^*_0 h_k+Q_k$ and $\tau^*_0+\sum_{k=1}^{K}\tau^*_k\leq T_{\max}$.  Then, we can construct another solution $\{\widetilde{P}_0,0,\{{\widetilde{p}}_{k}\},$ $\{ {\widetilde{\tau}}_k\}\}$ with $\widetilde{P}_0=P_0^*$,  $\widetilde{p}_k=p_k^*$, and $\widetilde{\tau}_k=\alpha\tau_k^*$, where $\alpha= \min\limits_{k \in \Phi_{\mathcal{I}}} \frac{Q_k}{Q_k+\eta P^*_0\tau^*_0h_k}\leq1$ such that
       $\widetilde{\tau}_k(\frac{\widetilde{p}_k}{\varsigma}+p_c)\leq Q_k$ for $\forall\, k$.  It can be verified that $\{\widetilde{P}_0,0,\{{\widetilde{p}}_{k}\}, \{ {\widetilde{\tau}}_k\}\}$ is a feasible point in $D_{\mathcal{I}}$, and can achieve the same EE as $\{P^*_0, \tau_0^*,\{{p}^*_{k}\}, \{{\tau}^*_k\}\}\in D$, i.e.,
   $
 \frac{ \sum_{ k\in\Phi_{\mathcal{I}}}{\widetilde{\tau}}_kr_k({\widetilde{p}}_k)}{\sum_{k\in\Phi_{\mathcal{I}}}{\widetilde{\tau}}_k(\frac{{\widetilde{p}}_k}{\varsigma}+p_{c})}=
  \frac{ \sum_{ k\in\Phi_{\mathcal{I}}}\alpha \tau^*_kr_k({p}^*_k)}{\sum_{k\in\Phi_{\mathcal{I}}}\alpha\tau^*_k(\frac{{p}^*_k}{\varsigma}+p_{c})}
  =\frac{ \sum_{ k\in\Phi_{\mathcal{I}}}\tau^*_kr_k(p^*_k)}{\sum_{k\in\Phi_{\mathcal{I}}}\tau^*_k(\frac{p^*_k}{\varsigma}+p_{c})}.$
On the other hand, since $\{{\check{P}}_0, 0,\{\check{p}_{k}\}, \{\check{\tau}_k\}\}\in D_{\mathcal{I}}$ maximizes $EE_{\rm{IELCN}}$, ``$b$'' holds true for the second term inside the  bracket.
\end{itemize}}

{The above analysis proves Theorem \ref{theorem0}.}

\section*{Appendix  B: Proof of Lemma \ref{theorem1}}\label{apdix1}
 We prove Lemma \ref{theorem1} by contradiction. Suppose that $\{P^*_{0},\{p^*_{k}\}, \tau^*_0,\{\tau^*_{k}\}\}$  is the optimal solution  to problem (\ref{eq132}) where $P^*_{0}<P_{\rm{\max}}$ holds for any $P^*_{0}$, and the optimal system EE is denoted as $EE^*$. Let $E^*_0\triangleq P^*_{0}\tau^*_0$ where $E^*_0$ can be interpreted as the actual energy transmitted by the power station. Then, we can construct another solution $\{\widetilde{P}_{0},\{{\widetilde{p}}_{k}\}, \widetilde{\tau}_0,\{{\widetilde{\tau}}_{k}\}\}$ satisfying
 $\widetilde{P}_{0} =P_{\rm{\max}}$, $\widetilde{P}_{0} \widetilde{\tau}_0=E^*_0$, $\widetilde{p}_{k}=p^*_{k}$, and $\widetilde{\tau}_k=\tau^*_{k}$, respectively.
 The corresponding system EE is denoted as $\widetilde{EE}$. It is easy to check that $\{\widetilde{P}_{0},\{{\widetilde{p}}_{k}\}, \widetilde{\tau}_0,\{{\widetilde{\tau}}_{k}\}\}$  is  a  feasible solution given $\{P^*_{0},\{p^*_{k}\}, \tau^*_0,\{\tau^*_{k}\}\}$. Moreover, since $\widetilde{P}_{0} =P_{\rm{\max}} >P^*_{0}$, it follows that $\widetilde{\tau}_0< \tau_0^*$ and hence $P_c\widetilde{\tau}_0<P_c\tau_0^* $ always holds true.
 Therefore, we always have
${ {\widetilde{P}}_{0}}\widetilde{\tau}_0\left(\frac{1}{\xi}-\sum_{k=1}^{K}\eta h_k \right)+P_c\widetilde{\tau}_0 <{P^*_{0}}\tau^*_0\left(\frac{1}{\xi}-\sum_{k=1}^{K}\eta h_k \right)+P_c\tau^*_0.$
   Since neither $\{p^*_{k}\}$ nor $\{\tau^*_{k}\}$ are changed in the constructed solution, based on  problem  (\ref{eq132}),  it follows that  $\widetilde{EE}>EE^*$, which contradicts  the assumption that $\{P^*_{0},\{p^*_{k}\}, \tau^*_0,\{\tau^*_{k}\}\}$ is the optimal solution. Lemma \ref{theorem1} is thus proved.

\section*{Appendix C: Proof of Lemma \ref{theorem2}}\label{apdix2}
 Suppose that $\{P^*_{0},\{p^*_{k}\}, \tau^*_0,\{\tau^*_{k}\}\}$ yields the maximum system EE,  $EE^*$,  and satisfies $0 \leq\tau^*_{0}+\sum_{k\in \Phi_{\mathcal{P}}}\tau^*_k< T_{\rm{\max}}$.
Then, we can construct another solution $\{\widetilde{P}_{0},\{\widetilde{p}_{k}\}, \widetilde{\tau}_0,\{\widetilde{\tau}_{k}\}\}$  with  $\widetilde{P}_{0}=P^*_{0}$, $\widetilde{p}_{k}=p^*_{k}$, $\widetilde{\tau}_0=\alpha  \tau^*_0$, $\widetilde{\tau}_{k}=\alpha \tau^*_{k}$, respectively, where $\alpha= \frac{T_{\rm{\max}}}{\tau^*_{0}+\sum_{k\in \Phi_{\mathcal{P}}}\tau^*_k}>1$ such that  $\widetilde{\tau}_{0}+\sum_{k\in \Phi_{\mathcal{P}}}\widetilde{\tau}_k= T_{\rm{\max}}$. The corresponding system EE is denoted as $\widetilde{EE}$. First,  it is easy to check that $\{\widetilde{P}_{0},\{\widetilde{p}_{k}\}, \widetilde{\tau}_0,\{\widetilde{\tau}_{k}\}\}$ still satisfies constraints C1-C5. Then, substituting  $\{\widetilde{P}_{0},\{\widetilde{p}_{k}\}, \widetilde{\tau}_0,\{\widetilde{\tau}_{k}\}\}$ into problem (\ref{eq132}) yields  $\widetilde{EE}= EE^*$, which means that the optimal system EE can always be achieved by using up all the available time, i.e., $T_{\rm{\max}}$. Lemma \ref{theorem2} is thus proved.

\section*{Appendix D: Proof of Lemma \ref{theorem3}}\label{apdix3}
{First, if $EE^*_{\rm{PWPCN}}<ee_m^{\star}$, we proved that user $m$ will be scheduled in our previous work [Theorem 1]\cite{qing15_scheduling}}. Second, we prove that the scheduled user will use up all of its energy by contradiction. Suppose that $\{P^*_{0},\{p^*_{k}\}, \tau^*_0,\{\tau^*_{k}\}\}$  is the optimal solution  to problem (\ref{eq132})  and  there  exists a $U_m$, $\forall\,  m \in \Phi_{\mathcal{P}}$, such that $EE^*_{\rm{PWPCN}}<ee^*_m$,  but its harvested energy is not used up, i.e., $(\frac{p^*_m}{\varsigma}+p_c)\tau^*_m <\eta P_{\rm{\max}}\tau^*_0h_m$ and $(\frac{p^*_k}{\varsigma}+p_c)\tau^*_k \leq \eta P_{\rm{\max}}\tau^*_0h_k$ for $k\neq m$. The corresponding system EE, $EE^*_{\rm{PWPCN}}$, is given by
\begin{align}
EE^*_{\rm{PWPCN}}= \frac{\sum_{k\neq m} \tau^*_kW\log_2\left(1+{p^*_k\gamma_k}\right)+ \tau^*_mW\log_2\left(1+{p^*_m\gamma_m}\right)}
{{P^*_{0}}\tau^*_0\left(\frac{1}{\xi}-\sum_{k=1}^{K}\eta h_k \right)+P_c\tau^*_0 + \sum_{k\neq m}\tau^*_k\left(\frac{p^*_k}{\varsigma}+p_{c}\right) +\tau^*_m \left(\frac{p^*_m}{\varsigma}+p_{c}\right)}.
\end{align}
Then, we can construct another solution $\{\widetilde{P}_{0},\{\widetilde{p}_{k}\}, \widetilde{\tau}_0,\{\widetilde{\tau}_{k}\}\}$ with $\widetilde{P}_{0}=P^*_{0}$, $\widetilde{p}_{k}=p^*_{k}$ for $\forall\, k$, $\widetilde{\tau}_0=\beta \tau^*_0$, $\widetilde{\tau}_{k}=\beta \tau^*_{k}$ for $k\neq m $,  and  $\widetilde{\tau}_{m}=\alpha \tau^*_{m}$, respectively, where $0<\beta<1$ and $\alpha>1$. Note that as $\beta\rightarrow 0$, it follows that $\eta P_{\rm{\max}}\widetilde{\tau}_0h_m=\beta \eta P_{\rm{\max}}\tau^*_0h_m \rightarrow 0$, and as $\alpha $ increases, it follows that  $(\frac{\widetilde{p}_m}{\varsigma}+p_c)\widetilde{\tau}_m=\alpha(\frac{p^*_m}{\varsigma}+p_c)\tau^*_m$ increases. Therefore, there always exist $\alpha$ and $\beta$ such that $\alpha(\frac{p^*_m}{\varsigma}+p_c)\tau^*_m = \beta \eta P_{\rm{\max}}\tau^*_0h_m$ holds.
{It is also easy to check that for $k\neq m$, $\beta(\frac{p^*_k}{\varsigma}+p_c)\tau^*_k \leq\beta\eta P_{\rm{\max}}\tau^*_0h_k$ still holds.}
Consequently, the corresponding system EE, denoted as $\widetilde{EE}_{\rm{PWPCN}}$, is given by
\begin{align}
\widetilde{EE}_{\rm{PWPCN}}&= \frac{\sum_{k\neq m} \widetilde{\tau}_kW\log_2\left(1+{\widetilde{p}_k\gamma_k}\right)+ \widetilde{\tau}_mW\log_2\left(1+{\widetilde{p}_m\gamma_m}\right)}
{{\widetilde{P}_{0}}\widetilde{\tau}_0\left(\frac{1}{{\xi}}-\sum_{k=1}^{K}\eta h_k \right)+P_c\widetilde{\tau}_0 + \sum_{k\neq m}\widetilde{\tau}_k(\frac{\widetilde{p}_k}{\varsigma}+p_{c}) +\widetilde{\tau}_m \left(\frac{\widetilde{p}_m}{\varsigma}+p_{c}\right)}\\
&= \frac{\beta \sum_{k\neq m} \tau^*_kW\log_2\left(1+{p^*_k\gamma_k}\right)+ \alpha \tau^*_mW\log_2\left(1+{p^*_m\gamma_m}\right)}
{\beta \left({P^*_{0}}\tau^*_0\left(\frac{1}{{\xi}}-\sum_{k=1}^{K}\eta h_k \right)+P_c\tau^*_0 + \sum_{k\neq m}\tau^*_k\left(\frac{p^*_k}{\varsigma}+p_{c}\right) \right)+\alpha \tau^*_m \left(\frac{p^*_m}{\varsigma}+p_{c}\right)}.\nonumber
\end{align}
In order to compare $EE^*_{\rm{PWPCN}}$ and $\widetilde{EE}_{\rm{PWPCN}}$, we introduce Lemma \ref{lma1}.
\begin{lemma}\label{lma1}
Assume that  $a$, $b$, $c$, and $d$ are arbitrary positive numbers which satisfy $\frac{a+c}{b+d}<\frac{c}{d}$. Then, for any $0<\beta<\alpha$, we always have
$\frac{a+c}{b+d}<\frac{\beta a+ \alpha c}{\beta b+\alpha d}.$
\end{lemma}
\begin{proof}
The proof is straightforward and thus omitted due to the space limitation.
\end{proof}

Let $a=\sum_{k\neq m} \tau^*_kW\log_2\left(1+{{p}^*_k\gamma_k}\right)$, $b={P^*_{0}}\tau^*_0(\frac{1}{\xi}-\sum_{k=1}^{K}\eta h_k )+P_c\tau^*_0 + \sum_{k\neq m}\tau^*_k(\frac{p^*_k}{\varsigma}+p_{c})$, $c=\tau^*_mW\log_2\left(1+{{p}^*_m\gamma_m}\right)$, and $d=\tau^*_m (\frac{p^*_m}{\varsigma}+p_{c})$, respectively. 
 {Since user $m$ is scheduled, we have $EE^*_{\rm{PWPCN}}<ee_m^*$, i.e., $\frac{a+c}{b+d}<\frac{c}{d}$, otherwise, $EE^*_{\rm{PWPCN}}$ can be further increased by letting $\tau^*_m=0$.  Based on Lemma \ref{lma1}, we obtain $EE^*_{\rm{PWPCN}}<\widetilde{EE}_{\rm{PWPCN}}$, which contradicts the assumption, and 1) in Lemma \ref{theorem3} is thus proved. The proofs of  2) and 3) can be obtained easily following a similar procedure as above, and thus are omitted here for brevity.}  
\section*{Appendix E: Proof of Theorem \ref{theorem4}}\label{apdix4}
Denote $S^*\subseteq \Phi_{\mathcal{P}}$ as the set of users which are scheduled. Substituting  $P_0=P_{\rm{\max}}$ and $\tau_k=\frac{\eta P_{\rm{\max}}h_k\tau_0}{\frac{p_k}{\varsigma}+p_c}$  into the objective function of problem (\ref{eq132}), we have
\begin{align}
EE&= \frac{\sum_{k\in S^*} \frac{\eta P_{\rm{\max}}h_k\tau_0}{\frac{p_k}{\varsigma}+p_c}W\log_2\left(1+p_k\gamma_k\right)}
{{P_{\rm{\max}}}\tau_0(\frac{1}{{\xi}}-\sum_{k=1}^{K}\eta h_k )+P_c\tau_0 + \sum_{k\in S^*}\frac{\eta P_{\rm{\max}}h_k\tau_0}{\frac{p_k}{\varsigma}+p_c}(\frac{p_k}{\varsigma}+p_{c})} \nonumber\\
&=\frac{\eta P_{\rm{\max}} \sum_{k\in S^*} h_kee_k}
{{P_{\rm{\max}}}(\frac{1}{{\xi}}-\sum_{k=1}^{K}\eta h_k )+P_c + \eta P_{\rm{\max}}\sum_{k\in S^*}h_k},
\end{align}
where $ee_k$ is the user EE defined in (\ref{eq11}). Given $S^*$, in order to maximize $EE$, we only have to maximize each $ee_k$, which is solely determined by $p_k$,  and the maximum value $ee^{\star}_k$ can be computed from (\ref{eq11}) and (\ref{eq12}). After some manipulations, we obtain
\begin{align}\label{apdx}
EE^*=\frac{\sum_{k\in S^*}h_k ee^{\star}_k}{\frac{1}{\eta\xi} \left(\frac{P_{c}}{ P_{\mathop{\max}}}\xi+1-\sum_{k=1}^{K}\xi\eta h_k\right)+ \sum_{k\in S^*}h_k}.
\end{align}
{Since the transmit power of each scheduled user $k$ is $p^{\star}_k$ given by (\ref{apdx}), $\tau^*_0$ and $\tau^*_k$ can be easily obtained from Lemma 2 and Lemma 3.} It is worth noting that there only exists a relationship between $\tau_0$ and $\tau_k$, $\forall\, k$ as in (\ref{eq15_3}). The value of $\tau_0$ can be scaled without affecting the system EE of PWPCN in the feasible region.
Theorem \ref{theorem4} is thus proved.
\section*{Appendix F: Proof of Theorem \ref{theorem01}}\label{apdix01}
From (\ref{eq133}),  we have
\begin{align}\label{apdx_eq02}
EE^*_{\rm{IELCN}}&= \frac{ \sum_{ k\in\Phi_{\mathcal{I}}}\tau^*_kW\log_2\left(1+{p^*_k\gamma_k}\right)}{\sum_{ k\in\Phi_{\mathcal{I}}}\tau^*_k(\frac{p^*_k}{\varsigma}+p_{c})}  \overset{c}\leq \max\limits_{k\in \Phi_{\mathcal{I}}}\frac{ \tau^*_kW\log_2\left(1+{p^*_k\gamma_k}\right)}{\tau^*_k(\frac{p^*_k}{\varsigma}+p_{c})}    \nonumber\\
&\overset{d}\leq \max\limits_{k\in \Phi_{\mathcal{I}}} \frac{ W\log_2\left(1+{p^{\star}_k\gamma_k}\right)}{\frac{p^{\star}_k}{\varsigma}+p_{c}}=ee^{\star}_k
\end{align}
where inequality ``$c$'' holds due to the same argument as inequality ``$a$'' in (\ref{apdx_eq01}), and ``$d$''  follows from the optimality of $p^{\star}_k$ for $ee^{\star}_k$.
From (\ref{apdx_eq02}), we observe that the maximum system EE is always achieved by scheduling a single user. Then, applying the optimal power $p^{\star}$ in the time and energy harvesting constraints, we obtain (\ref{eq135}) and (\ref{eq136}). Similarly, the value of $\tau_k$ $\forall\, k$ does not affect the system EE of IELCN.
\section*{Appendix G: Proof of Theorem \ref{theorem70}}\label{apdix70}
By taking the partial derivative of $\mathcal{L}$ with respect to $\tau_0$, $E_k$, and $\tau_k$, respectively, we  obtain
\begin{align}
\frac{\partial\mathcal{L}}{\partial \tau_0}&= P_{\max}\left(\sum_{k=1}^{K}\mu_kh_k-q\left(\frac{1}{\xi}-\sum_{k=1}^{K}\eta h_k\right)\right)-qP_c-\delta, \label{apdx_eq20}\\
\frac{\partial\mathcal{L}}{\partial E_k}&=\frac{W(1+\vartheta)\tau_k\gamma_k}{(\tau_k+E_k\gamma_k)\ln2}-\frac{q+\mu_k}{\varsigma}, \label{apdx_eq21}\\
\frac{\partial\mathcal{L}}{\partial \tau_k}&=W(1+\vartheta)\log_2\left(1+\frac{E_k}{\tau_k}\gamma_k\right)-
                                                                                              \frac{W(1+\vartheta)E_k\gamma_k}{(\tau_k+E_k\gamma_k)\ln2}-(q+\mu_k)p_c-\delta, \label{apdx_eq22}
\end{align}
and the complementary slackness conditions are given by
\begin{align}
\mu_k\left(Q_k+ \eta P_{\max}\tau_0h_k-\frac{E_k}{\varsigma}-p_{c}\tau_k\right)&=0,   \label{apdx_eq200} \\
\delta\left( T_{\max}-\tau_0-\sum_{k=1}^{K}\tau_k\right)& =0,                \label{apdx_eq201}\\
\vartheta\left(\sum_{k=1}^{K}\tau_kW\log_{2}\left(1+\frac{E_{k}}{\tau_k}\gamma_{k}\right)-R_{\min} \right)&=0.     \label{apdx_eq202}
\end{align}

 Let $f_0(\bm{\mu})\triangleq\frac{\partial\mathcal{L}}{\partial \tau_0}$ and $f(\gamma_k,\mu_k)\triangleq\frac{\partial\mathcal{L}}{\partial \tau_k}$.  From (\ref{apdx_eq20}), we know that $\mathcal{L}$ is a linear function of $\tau_0$. Since $\tau_0\geq 0$, to make sure that the Lagrangian function $\mathcal{L}$ is bounded above \cite{Boyd}, we have $f_0(\bm{\mu})\leq0$. Specifically, when $f_0(\bm{\mu})<0$, it follows that $\tau_0=0$, otherwise if $f_0(\bm{\mu})=0$, $\tau_0\geq0$, which results in (\ref{solu_1}). From $\frac{\partial\mathcal{L}}{\partial E_k}=0$, we can obtain the relationship between $E_k$ and $\tau_k$ as
 \begin{align}\label{apdx_eq23}
 p_k=\frac{E_k}{\tau_k}=\left[\frac{W(1+\vartheta)\varsigma}{(q+\mu_k)\ln2}-\frac{1}{\gamma_k}\right]^+, \forall \, k.
 \end{align}
 Substituting  (\ref{apdx_eq23}) into (\ref{apdx_eq22}) and after some manipulations, $f(\gamma_k,\mu_k)$ can be expressed as
 \begin{align}\label{apdx_eq24}
f(\gamma_k,\mu_k)\nonumber &= (1+\vartheta)W\log_2\left(1+\gamma_k\left[\frac{W(1+\vartheta)\varsigma}{(q+\mu_k)\ln2}-\frac{1}{\gamma_k}\right]^+\right) \nonumber \\
&-(q+\mu_k)\left(\left[\frac{W(1+\vartheta)\varsigma}{(q+\mu_k)\ln2}-\frac{1}{\gamma_k}\right]^{+}+p_{c}\right)-\delta.
 \end{align}
 Since $\tau_k\geq 0$, using a similar analysis as for $\tau_0$, the optimal solution of $\tau_k$ must satisfy
 \begin{align}\label{apdx_eq25}
 \frac{\partial\mathcal{L}}{\partial \tau_k}= f(\gamma_k,\mu_k) \left\{
\begin{array}{lcl}
< 0,&\tau_k=0,\textcolor{white}{\,\,~\forall k}\\
= 0,&\tau_k\geq 0, ~\forall k.
\end{array}\right.
 \end{align}
  To facilitate our derivation, we next introduce a lemma related to  $f(\gamma_k,\mu_k)$.
 \begin{lemma}\label{lemma5}
$f(\gamma_k,\mu_k)$ is an increasing function of $\gamma_k$ and   a decreasing function of $\mu_k$ under the condition that $\frac{W(1+\vartheta)\varsigma}{(q+\mu_k)\ln2}>\frac{1}{\gamma_k}$.
 \end{lemma}
 \begin{proof}
Lemma 6 can be easily proved by taking the derivative of $f(\gamma_k,\mu_k)$ with respect to $\gamma_k$ and $\mu_k$, respectively. The proof is thus omitted due to the space limitation.
\end{proof}

 Based on Lemma \ref{lemma5}, we know that the maximum value of $f(\gamma_k,\mu_k)$ in terms of $\mu_k$ is achieved at $\mu_k=0$, i.e., $ f(\gamma_k,0)$. Moreover,  when $\gamma_k= \frac{q\ln2}{W(1+\vartheta)\varsigma}$, $f(\gamma_k,0)=-qp_c-\delta<0$ holds, and when $\gamma_k\rightarrow +\infty$, $f(\gamma_k,0)\rightarrow +\infty$ holds. { From (\ref{apdx_eq24}), since $f(\gamma_k,0)$ is an increasing function of $\gamma_k$, there always  exists a $x^*$ such that $f(x^*,0)=0$, i.e.,
\begin{align} \label{eq4.9}
f(x^*,0)= (1+\vartheta)W\log_2\left(\frac{W(1+\vartheta)x^*}{q\ln2}\right)+\frac{q}{x}
-\frac{W(1+\vartheta)\varsigma}{\ln2}-qp_{c}-\delta=0,
\end{align}
 which results in (\ref{solu_50}).}  Note that since the parameters $\vartheta$, $W$, $q$, $\varsigma$, $p_c$, and $\delta$ in (\ref{eq4.9}) do not depend on the user index $k$, the threshold $x^*$ is thereby identical for all users.
Now, we analyze the following three cases:
  \begin{itemize}
    \item  For $\gamma_k<x^*$, it follows that  $f(\gamma_k,\mu_k)\leq f(\gamma_k,0)<0$. According to (31), we know that a user with UL channel gain $\gamma_k$ less than $x_k$ is allocated zero transmission time, i.e., $\tau_k=0$.
    \item  For $\gamma_k>x^*$, there always exists a $\mu_k>0$ such that  $f(\gamma_k,\mu_k)=0 <f(\gamma_k,0)$ since $f(\gamma_k,\mu_k)$ is a decreasing function with respect to $\mu_k$. However, there may exist  $\mu_k>0$ such that $f(\gamma_k,\mu_k)<0$. Then, according to (29), it follows that  $\tau_k=0$ and $\frac{E_k}{\varsigma}+p_c\tau_k=0<\eta P_{\max}\tau^{*}_0h_k+Q_k$, which contradicts  (\ref{apdx_eq200}), i.e., $\mu_k(\eta P_{\max}\tau^{*}_0h_k+Q_k-\frac{E_k}{\varsigma}-p_{c}\tau_k)=0$, and this is thereby not the optimal solution. Nevertheless, for users with $\gamma_k$ larger than $x^*$, $\mu_k>0$ implies that they utilize all of their energy. Thus, from (\ref{apdx_eq200}), we have $\tau_k= \frac{\eta P_{\max}\tau^{*}_0h_k+Q_k}{\frac{p_k}{\varsigma}+p_c}$. {Correspondingly, as $\tau_k>0$, the value of $\mu_k$ can be calculated from the second case in (\ref{apdx_eq25}), where $f(\gamma_k,\mu_k)$ is given by (\ref{apdx_eq24}), i.e.,  $f(\gamma_k,\mu_k)=0$.
        }
    \item For $\gamma_k=x^*$, if $\mu_k>0$, then $f(\gamma_k,\mu_k)=0 <f(\gamma_k,0)=0$ and $\tau_0=0$, which contradicts  (\ref{apdx_eq200}). Therefore, $\mu_k=0$ follows from  (\ref{apdx_eq200}), this means that user $k$ can utilize any portion of its energy, i.e., $\tau_k \in\left[0,\frac{\eta P_{\max}\tau^{*}_0h_k+Q_k}{\frac{p_k}{\varsigma}+p_c}\right]$.
  \end{itemize}
Based on the above three cases, we obtain the region of time allocation variables given in (\ref{solu_1}) and (\ref{solu_3}). As the  Lagrangian function $\mathcal{L}$ is a linear function of $\tau_0$ and $\tau_k$, the maximum value of $\mathcal{L}$ can always be obtained at the vertices of the region created by (\ref{solu_1}) and (\ref{solu_3}). Moreover,  $\tau^{*}_0$ and $\tau_k^{*}$, for $k=1,...,K$,  satisfy the complementary slackness conditions (\ref{apdx_eq201}) and (\ref{apdx_eq202}).
{Therefore, if $\delta>0$, then the time  constraint should be strictly met with equality, otherwise, we obtain an associated inequality for limiting the range of  time variables $\tau_k$ and $\tau_0$. 
The same interpretation also applies to $\vartheta$.


\section*{Appendix H: Proof of Corollary \ref{theorem7}}\label{apdix7}
From (\ref{apdx_eq24}) and (\ref{apdx_eq25}), we know that for each scheduled user $k$, we have
\begin{align}\label{apdx_eq4}
 (1+\vartheta)W\log_2\left(1+\gamma_kp_k\right) &-(q+\mu_k)\left(\frac{p_k}{\varsigma}+p_{c}\right)-\delta=0.
\end{align}
Note that from (\ref{apdx_eq23}),  $\frac{W(1+\vartheta)\varsigma}{(q+\mu_k)\ln2}= p_k+\frac{1}{\gamma_k}$ also holds for user $k$. Substituting this relation into (\ref{apdx_eq4}) and after some manipulations, we  obtain
\begin{align}\label{apdx_eq6}
\mathcal{D}(p_k)\triangleq {W\log_2\left(1+p_k\gamma_k\right) }-\frac{W\varsigma}{\left(p_k+\frac{1}{\gamma_k}\right)\ln2}\left(\frac{p_k}{\varsigma}+p_{c}\right)- \frac{\delta}{1+\vartheta}=0.
\end{align}
If the total available transmission time is not used up, i.e., $\tau_0+\sum_{k=1}^{K}\tau_k< T_{\max}$, it follows from (\ref{apdx_eq201}) that $\delta=0$.
Note that $\mathcal{D}(p_k)$ is increasing in $p_k$. Moreover, when $p_k=0$,   $\mathcal{D}(p_k)=-\frac{W\varsigma\gamma_k}{\ln2}p_c<0$, and when $p_k\rightarrow +\infty$,  $\mathcal{D}(p_k)\rightarrow +\infty$.  Therefore, there is  always a unique solution $p_k$ for (\ref{apdx_eq6}).  Combining (\ref{apdx_eq6}) with (\ref{eq11}) and (\ref{eq12}), after some manipulations,  we conclude $p_k=p^{\star}_k$. On the other hand, if the  total available transmission time is used up, i.e., $\tau_0+\sum_{k=1}^{K}= T_{\max}$, it follows that $\delta\geq 0$. {Hence,  we conclude that $p_{k} \geq p^{\star}_k$ since $\mathcal{D}(p_k)$ is monotonically increasing with respect to $p_k$.}  Corollary \ref{theorem7} is thus proved.

\section*{Appendix I: Proof of Corollary \ref{theorem8}}\label{apdix8}
If WET is activated, i.e., $\tau_0>0$, from (\ref{solu_1}), we obtain
\begin{align}\label{apdx_eq7}
\delta= P_{\max}\sum_{k=1}^{K}(q+\mu_k)h_k-q\left(\frac{P_{\max}}{\xi}+P_c\right).
\end{align}
Meanwhile, for any scheduled user $k$, it follows that  (\ref{apdx_eq6}) also holds true.
Combining   (\ref{apdx_eq7}) and (\ref{apdx_eq6}), and after some manipulations, we obtain
\begin{align}\label{apdx_eq8}
W\log_2(1+p_k\gamma_k)-\frac{W(p_k+p_c\varsigma)}{(p_k+\frac{1}{\gamma_k})\ln2}-\sum_{k=1}^K\frac{WP_{\max}\varsigma h_k}{(p_k+\frac{1}{\gamma_k})\ln2}+\frac{q\left(\frac{P_{\max}}{\xi}+p_c\right)}{1+\vartheta}=0.
\end{align}

If the energy of any user $m$ is not used up, $\mu_m=0$ holds due to the associated
complementary slackness condition in (\ref{apdx_eq200}). Thus, from (\ref{apdx_eq23}), we know that
$p_m=\frac{E_m}{\tau_m}=\left[\frac{W(1+\vartheta)\varsigma}{q\ln2}-\frac{1}{\gamma_m}\right]^+, \forall\, m.$
Therefore, substituting $p_m$ into (\ref{apdx_eq8}), we have
\begin{align}\label{apdx_eq10}
\log_2(1+p_k\gamma_k)-\frac{(p_k+p_c\varsigma)}{(p_k+\frac{1}{\gamma_k})\ln2}-\sum_{k=1}^K\frac{P_{\max}\varsigma h_k}{(p_k+\frac{1}{\gamma_k})\ln2}+\frac{\varsigma\left(\frac{P_{\max}}{\xi}+p_c\right)}{(p_m+\frac{1}{\gamma_m})\ln2}=0.
\end{align}
From (\ref{apdx_eq10}), we observe that the transmit powers of the scheduled users depend only on the system parameters $W$, $\xi$, $\varsigma$, $P_{\max}$,  $p_c$, and $\gamma_k$.  Moreover, the left hand side of (\ref{apdx_eq10}) is a monotonically increasing  function of $p_k$. Therefore, as long as the energy of user $m$ is not used up, (\ref{apdx_eq10}) holds true and $p_k$ remains constant. Note that if the energy of user $m$ is used up, i.e., $\mu_m>0$, $p_m=\left[\frac{W(1+\vartheta)\varsigma}{(q+\mu_m)\ln2}-\frac{1}{\gamma_m}\right]^+$ and $\mu_m$ is thereby introduced in (55).  Then, the value of $p_k$ varies with $\mu_m$.
On the other hand, since WET is used, i.e., $\tau_0>0$, and the energy of user $m$ is not used up, it can be further  shown that the total available time must be used up, i.e., $\tau_0+\sum_{k=1}^{K}\tau_k=T_{\max}$. At the same time, the required system throughput has to be satisfied, i.e., $\sum_{k=1}^{K}\tau_k\log_2(1+p_kg_k)\geq R_{\min}$.
Therefore, as $R_{\min}$ increases, the information transmission time $\sum_{k=1}^{K}\tau_k$ has to be increased since  $p_k$ remains constant. Thus, it follows that $\tau_{0}$ decreases due to the more stringent time constraint. Then, the energy harvested at each user $\eta P_{\max}h_k\tau_0$ decreases and the transmission time for any user $k\neq m$ also decreases as $\tau_k =\frac{P_{\max}h_k\tau_0+Q_k}{\frac{p_k}{\varsigma}+p_c}$.

\bibliographystyle{IEEEtran}
\bibliography{IEEEabrv,mybib}

\end{document}